\documentclass[conference]{IEEEtran}
\IEEEoverridecommandlockouts
\usepackage{cite}
\usepackage{setspace}
\usepackage{amsmath,amssymb,amsfonts, amsthm}
\usepackage{algorithmic}
\usepackage{graphicx}
\usepackage{textcomp}
\usepackage{xcolor}
\usepackage{bbm}
\usepackage{physics}
\usepackage{subcaption}
\usepackage{mathtools}
\newtheorem{theorem}{Theorem}
\newtheorem{proposition}{Proposition}
\newtheorem{definition}{Definition}
\newtheorem{corollary}{Corollary}
\newtheorem{lemma}{Lemma}
\theoremstyle{remark}
\newtheorem{remark}{Remark}
\def\BibTeX{{\rm B\kern-.05em{\sc i\kern-.025em b}\kern-.08em
    T\kern-.1667em\lower.7ex\hbox{E}\kern-.125emX}}

\usepackage{fancyhdr}
\allowdisplaybreaks
\begin{document}
\onecolumn
\setstretch{1.5}
\title{Communication Complexity of Exact Sampling under R\'enyi Information 
\thanks{The authors are with the Department of Mathematics and Statistics, Queen's University, Kingston, Ontario, Canada. This work was presented in part at the Recent Results Session of the 2025 IEEE International Symposium on Information Theory and the 61st Allerton Conference on Communication, Control, and Comp     uting. This work was supported by the Natural Sciences and Engineering Research Council of Canada. Email: \{spencer.hill, fa, tamas.linder\}@queensu.ca.}}

\author{Spencer Hill, Fady Alajaji, and Tam\'as Linder}

\maketitle
\thispagestyle{plain}
\begin{abstract} 
We study the problem of exact sampling under an exponential communication cost, specifically Campbell's average codeword length $L(t)$ of order $t$~\cite{campbell1965coding}, and R\'enyi's entropy. We provide a lower bound on the Campbell cost of exact sampling that grows approximately as $D_{\frac{1}{\alpha}}(P||Q)$, the R\'enyi divergence of order $1/\alpha$, with $\alpha = \frac{1}{1+t}$. Using the Poisson functional representation of Li and El Gamal~\cite{sfrl}, we prove an upper bound on $L(t)$ whose leading R\'enyi divergence term has order within $\epsilon$ of that of the lower bound. Our results reduce to the bounds of Harsha et al.~\cite{harsha2010communication} as $\alpha \to 1$. We also provide numerical examples comparing the bounds in the cases of normal and Laplacian distributions, demonstrating that the upper and lower bounds are typically within 5-10 bits of each other. Our results characterize exactly the optimal asymptotic Campbell cost $L(t)$ per sample as the number of independent and identically distributed (i.i.d.) samples grows to infinity. We show that under the exponential cost, any causal sampler performs strictly worse asymptotically than noncausal samplers. This contrasts with the case of expected message length, where both causal and noncausal samplers have the same optimal asymptotic cost. 
\end{abstract}
\section{Introduction}
We consider the problem of communicating a sample from a probability distribution $P$, given that the sender and receiver have access to a shared source of randomness $\qty{U_i}_{i \geq 1}$ drawn independently according to another probability distribution $Q$. This problem, called exact sampling, is visualized in Fig.~\ref{fig:blockdiagram}. 
\begin{figure*}[ht]
    \centering
    \includegraphics{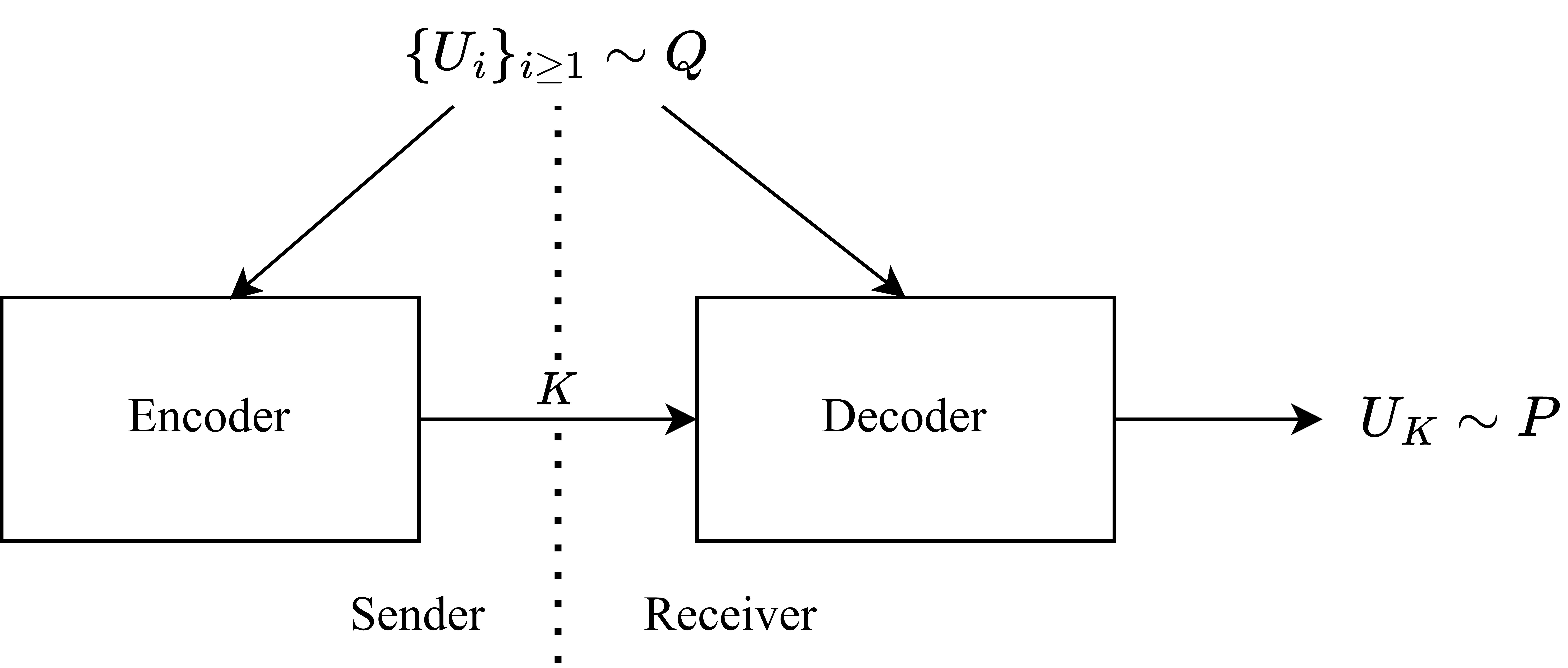}
    \caption{Block diagram of the exact sampling problem.}
    \label{fig:blockdiagram}
\end{figure*}
As studied in~\cite{harsha2010communication, braverman2014public, sfrl}, among others, the key question is how to design a communication protocol such that the expected message length needed to communicate the sample from $P$ is minimized. Naively communicating a sample from $P$ using lossless coding is often impossible, as in principle, the sample can have a continuous distribution. Instead, protocols aim to transmit an index $K \in \mathbb{N}$, where $\mathbb{N} = \qty{1, 2,\ldots}$ is the set of positive integers, such that the $K$th sample in the shared randomness has exact distribution $P$, which we write as $U_K \sim P$. 

The goal is then to minimize the expected message length of $K$, or (almost) equivalently, to minimize $H(K)$, where $H(\cdot)$ denotes Shannon entropy. Formally, with $\mathcal{C} : \mathbb{N} \rightarrow \qty{0,1}^*$ denoting a uniquely decodable binary variable-length code and $l(\mathcal{C}(k))$ denoting the length of the codeword of~$k$, the goal is to design a sampling algorithm and code $\mathcal{C}$ such that $\mathbb{E}[l(\mathcal{C}(K))]$ is minimized. This is intimately related to the problem of channel simulation (visualized in Fig.~\ref{fig:channelsimulationblock}), where $X$ and $Y$ are general random variables with joint distribution $\text{P}_{XY}$ and we let $P = \text{P}_{Y|X}(\,\cdot\, \mid x)$ and $Q = \text{P}_Y$; i.e., for input $x$ drawn via $\text{P}_X$, we wish to sample from the conditional distribution $\text{P}_{Y \mid X}(\, \cdot \, \mid x)$ given shared access to the marginal distribution $\text{P}_Y$. 
\begin{figure*}[ht]
    \centering
    \includegraphics{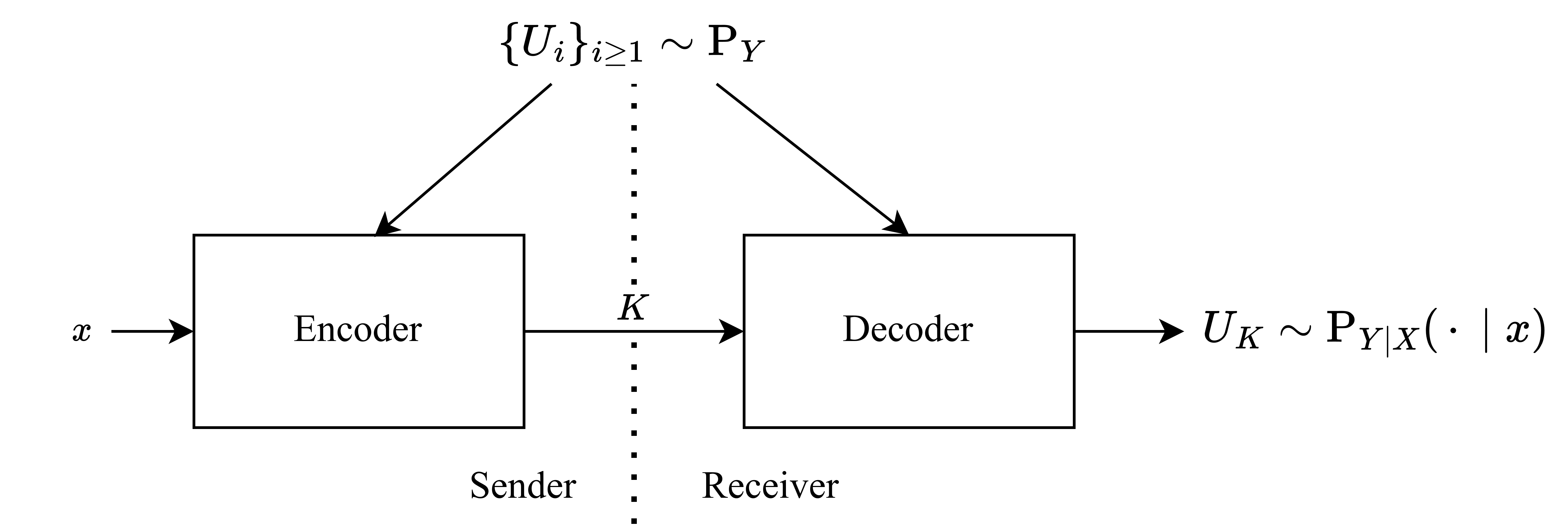}
    \caption{Block diagram of the channel simulation problem using sampling.}
    \label{fig:channelsimulationblock}
\end{figure*}
Channel simulation, also called reverse channel coding or relative entropy coding, has recently gained significant theoretical and practical attention due to its promise as an alternative for quantization in deep learning-based compression systems~\cite{flamich2020compressing, theis2022algorithms}. There has been a corresponding research effort to understand the theoretical limits of channel simulation and its related problems~\cite{goc2024channel, sriramu2024optimal}. 

Generalizing, one may also be concerned about costs other than the expected message length.  Motivated by Campbell's lossless source coding theorem~\cite{campbell1965coding}, we consider one-shot exact sampling under an exponential storage cost. Such a cost is appropriate for applications where buffer overflow can occur~\cite{jelinek1968buffer, 1056322,blumer1988renyi, somazzi2024nonlinear}, and therefore, the cost of long codewords is especially significant~\cite{courtade2014cumulant, courtade2014variable, saito2022non}. Here for a given $t > 0$, we aim to minimize the average codeword length of order $t$, also called the normalized cumulant generating function of codeword lengths~\cite{courtade2014cumulant, courtade2014variable, saito2022non}, given by
\begin{equation}
    L(t) = \frac{1}{t} \log(\mathbb{E}[2^{tl(\mathcal{C}(K))}]). \label{eq:campbellcostoriginal}
\end{equation}
As $t \rightarrow 0$ in~\eqref{eq:campbellcostoriginal}, by L'H\^opital's rule, we recover the average codeword length, $\mathbb{E}[l(\mathcal{C}(K))]$. In this work, we provide R\'enyi generalizations of prior results to the Campbell cost $L(t)$. In particular, we lower bound $L(t)$ for any sampling algorithm and use the Poisson functional representation to prove an upper bound on $L(t)$ whose leading R\'enyi divergence term has order within $\epsilon$ of the lower bound. We also describe an operational procedure for encoding $K$ with a universal code and prove an upper bound on $L(t)$ that reduces to the bound of Harsha et al.~\cite{harsha2010communication} as $t \to 0$. Qualitatively, the lower and upper bounds match in shape and are within 5-10 bits of each other for most $0 < \alpha = \frac{1}{1+t} < 1$, even for distributions where $L(t)$ is very large. Our bounds admit closed forms for many distributions~\cite{gil2013renyi} and are otherwise easily computed numerically, unlike the natural bound $L(t) < H_\alpha(K) + 1$ (Proposition~\ref{prop:campbellcost}), whose right-hand side is computable only for special $P, Q$ under the Poisson functional representation but quickly becomes intractable as $L(t)$ grows. We also analyze the Campbell cost in the asymptotic regime, where the goal is for the receiver to sample from the product distribution $P^{\otimes n}$ after receiving a single message. We derive exactly the optimal rate $\lim_{n \to \infty} L(t) / n$ over all sampling algorithms, and show that under mild assumptions, any causal sampler~\cite{liu2018rejection} (which examines the shared sequence one at a time) does strictly worse than a noncausal sampling algorithm. This differs from the case of the expected message length, where noncausal and causal algorithms achieve the same optimal asymptotic communication rate. 
\par
The remainder of the paper is organized as follows. Section~\ref{sec:preliminaries} contains some preliminary definitions and a review of the relevant literature, as well as a detailed summary of our main contributions. Section~\ref{sec:lowerbounds} lower bounds $L(t)$ in the exact sampling problem, and Section~\ref{sec:upperbounds} describes two communication protocols using the Poisson functional representation which closely approximate this lower bound. Section~\ref{sec:asymptotics} extends our results to the asymptotic regime, and Section~\ref{sec:simulations} gives numerical comparisons of the lower and upper bounds in the cases of normal and Laplacian distributions and a discussion of the difficulty of computing $H_\alpha(K)$.

\textit{Notation:} For $P, Q$ probability distributions, we write $P \ll Q$ to indicate that $P$ is absolutely continuous with respect to $Q$, in which case $\frac{\dd P}{\dd Q}$ denotes the Radon-Nikodym derivative. For a distribution $P$ on $\mathcal{U}$, we use $P^{\otimes n}$ to denote the product distribution on $\mathcal{U}^n$ corresponding to $n$ independent samples from $P$. We let $\qty{0, 1}^*$ denote the collection of all finite-length binary words, i.e., $\qty{0, 1}^* = \qty{0, 1, 00, 01, 10, 11, 000, \ldots}$. We write $\mathcal{N}(\mu, \sigma)$ for the normal distribution with mean $\mu$ and variance $\sigma^2$. Throughout, $\log$ is in base 2, all information measures are in bits, $\ln$ denotes the natural logarithm, and $\Gamma$ is the gamma function.

\section{Preliminaries and Contributions Summary} \label{sec:preliminaries}
\subsection{R\'enyi Information Measures}
Let $X$ be a discrete random variable with alphabet $\mathcal{X}$ and distribution $\text{P}_X$. For $\alpha \in (0, 1) \cup (1, \infty)$, we define the R\'enyi entropy of order $\alpha$ of $X$ as~\cite{renyi1961measures}
\begin{equation}
    H_\alpha(X) = \frac{1}{1-\alpha}\log(\sum_{x \in \mathcal{X}} \text{P}_X(x)^\alpha). \label{eq:Renyientropy}
\end{equation}
$H_\alpha(X)$ is nonincreasing in $\alpha$, and if $H_\alpha(X) < \infty$ for some $0 < \alpha < 1$, then $\lim_{\alpha \nearrow 1} H_\alpha(X) = H(X)$, recovering the Shannon entropy. In~\cite{campbell1965coding}, Campbell provided an operational meaning to the R\'enyi entropy by connecting it with the average code length of order $t$, $L(t)$. 
\begin{proposition}[\!\cite{campbell1965coding}]
\label{prop:campbellcost}
    Suppose the discrete random variable $X$ is encoded using a uniquely decodable binary code $\mathcal{C} : \mathcal{X} \rightarrow \qty{0, 1}^*$ such that each codeword has length $l(\mathcal{C}(x)) = n_x$, for $x \in \mathcal{X}$. Then the average code length of order $t>0$ given by
\begin{equation}
    L(t) = \frac{1}{t} \log(\sum_{x \in \mathcal{X}} \textup{P}_X(x) 2^{tn_x}), \label{eq:campbellcostdefinition}
\end{equation}
satisfies
\begin{equation}
    H_\alpha(X) \leq L(t), \label{eq:prop1lower}
\end{equation}
where $\alpha = \frac{1}{t + 1}$. Moreover, there exists a uniquely decodable binary code $\mathcal{C}$ such that
\begin{equation}
    L(t) < H_\alpha(X) + 1. \label{eq:prop1upper}
\end{equation}
\end{proposition} 
We note that while Campbell's original proof was for $X$ with finite alphabet, it can easily be extended to the case of a countably infinite alphabet. For simplicity, we will often refer to $L(t)$ as just the Campbell cost. $L(t)$ is a nondecreasing function of $t$, and is strictly increasing in $t$ unless all the lengths $n_x$ are equal. As a result, if $X$ has an infinite support, then $L(t) \to \infty$ as $t \to \infty$. If $X$ is finite-valued, then $L(t) \to \max_{x \in \mathcal{X}} n_x$ as $t \to \infty$. 

R\'enyi also defined the divergence of order $\alpha$ between probability distributions $P$ and $Q$, $D_\alpha(P||Q)$, as
\begin{equation}
    D_\alpha(P||Q) =
        \frac{1}{\alpha - 1} \log(\mathbb{E}_{X \sim Q} \left[\left(\frac{\dd P}{\dd Q}(X) \right)^\alpha\right]), \nonumber
\end{equation}
 where $\alpha \in (0, 1) \cup (1, \infty)$ and it is assumed that $P \ll Q$. We note that $\lim_{\alpha \nearrow 1} D_\alpha(P||Q) = D(P||Q)$, i.e., the Kullback-Leibler divergence is recovered when $\alpha \nearrow 1$~\cite{van2014renyi}. There are several R\'enyi generalizations of the Shannon mutual information~\cite{verdu2015alpha}; in our problem of channel simulation, Sibson's definition~\cite{sibson1969information} naturally emerges. For arbitrary random variables $X$ and $Y$ on Polish spaces $\mathcal{X}$ and $\mathcal{Y}$, with joint distribution $\text{P}_{XY}$ and marginals $\text{P}_X$ and $\text{P}_Y$, we define the Sibson $\alpha$-mutual information for any $\alpha \in (0, 1) \cup (0,\infty)$ as
 \begin{equation}
     I_\alpha^{\text{\scriptsize s}}(X;Y) = \inf_{Q_Y} D_\alpha(\text{P}_{XY} || \text{P}_X \times Q_Y). \label{eq:sibsondefinition}
 \end{equation}
 The infimum is achieved by the distribution $Q_{Y_\alpha}$ on $\mathcal{Y}$ defined by
 \begin{equation}
     \frac{\dd Q_{Y_\alpha}}{\dd \text{P}_Y}(y) = \frac{\mathbb{E}_X \left[ \left( \frac{\dd \text{P}_{XY}}{\dd \text{P}_X \text{P}_Y} (X, y) \right)^\alpha\right]^{\frac{1}{\alpha}}}{\mathbb{E}_Y \left[ \mathbb{E}_X \left[ \left( \frac{\dd \text{P}_{XY}}{\dd \text{P}_X \text{P}_Y} (X, Y) \right)^\alpha\right]^{\frac{1}{\alpha}}\right]}. \label{eq:sibsondistribution}
 \end{equation}
 As $\alpha \to 1$, $\lim_{\alpha \nearrow 1} I_\alpha^{\text{\scriptsize s}}(X; Y) = I(X;Y)$, recovering Shannon's mutual information.
 \subsection{Problem Setup}
Let $P$ and $Q$ be probability distributions on the Polish space $\mathcal{U}$ with $P \ll Q$. Suppose the sender and receiver have shared access to the i.i.d. sequence $\qty{U_i}_{i \geq 1}$ distributed according to $Q$; this sequence is referred to as \textit{common randomness}. The sender transmits an index $K$, encoded using the uniquely decodable binary variable-length code $\mathcal{C} : \mathbb{N} \to \qty{0, 1}^*$ with codeword lengths $l(\mathcal{C}(K))$, such that $U_K \sim P$. We refer to $Q$ as the \textit{proposal} distribution and $P$ as the \textit{target} distribution; it is always assumed that $P \ll Q$. This problem will henceforth be referred to as \textit{exact sampling} with proposal $Q$ and target~$P$. In this setup, we wish to minimize the Campbell cost, $L(t)$, of the message, $\mathcal{C}(K)$, over all protocols such that $U_K \sim P$. We also study the \textit{asymptotic} setting, where the proposal and target distributions are the product distributions $Q^{\otimes n}$ and $P^{\otimes n}$, respectively. For any $t > 0$, let $L^*_n(t)$ be the infimum of the set of achievable Campbell costs of order $t$ among all protocols for proposal $Q^{\otimes n}$ and target $P^{\otimes n}$. The goal is to characterize the optimal asymptotic rate $\limsup_{n \to \infty} L^*_n(t) / n$~\cite{li2024channel}.
\par
An important subclass of sampling algorithms is causal samplers~\cite{liu2018rejection}. A sampling algorithm is called causal if the index~$K$ is a stopping time, i.e., the event $\qty{K \geq k}$ is independent of $\qty{U_{k+1}, U_{k + 2}, \ldots}$. Intuitively, a causal sampler examines the shared randomness $\qty{U_i}_{i \geq 1}$ sequentially. Any sampler which is not causal is called noncausal. Examples of causal samplers include greedy rejection sampling~\cite{harsha2010communication,flamich2023adaptive} and greedy Poisson rejection sampling~\cite{flamich2023greedy}, while the Poisson functional representation~\cite{sfrl} is a noncausal sampler. 
\subsection{Prior Work}
The problem of channel simulation was first studied by Wyner~\cite{wyner1975common}, who considered the asymptotic version without common randomness but allowing error between the simulated and target distributions. Winter~\cite{winter2002compression} showed that, by allowing a shared source of unlimited randomness, it is possible to simulate the channel with small error at an asymptotic cost equal to the mutual information. Bennett et al.~\cite{bennett2002entanglement} showed that the same result holds even for exact channel simulation. More recently, Cuff~\cite{cuff2013distributed}, Bennett et al.~\cite{bennett2014quantum}, and Yu and Tan~\cite{yu2019exact} studied the tradeoff between communication rate and amount of common randomness in the asymptotic regime. Sriramu and Wagner~\cite{sriramu2024optimal} and Flamich et al.~\cite{flamich2025redundancy} examined the optimal rate of exact channel synthesis in the asymptotic case. There has also been recent work, motivated by applications in deep learning-based compression, on developing algorithms which are computationally efficient, e.g.~\cite{theis2022algorithms, flamich2022fast, flamich2023faster, sriramu2024fast, flamich2023greedy, flamich2023adaptive
}. We refer the reader to~\cite{li2024channel} for a comprehensive survey of channel simulation methods and results. \par 
In the problem of exact sampling, it is not hard to show that a lower bound on the expected message length is $\mathbb{E}[l(\mathcal{C}(K))] \geq D(P||Q)$~\cite{harsha2010communication}. The first achievability bound is due to Harsha et al.~\cite{harsha2010communication}, who described a communication protocol for discrete $P,Q$ for which 
\begin{equation}
    \mathbb{E}[l(\mathcal{C}(K))] \leq D(P||Q) + (1 + \epsilon)\log(D(P||Q) + 1) + c_\epsilon, \nonumber
\end{equation}
with $c_\epsilon$ a constant depending on $\epsilon$. Their proof effectively amounts to showing that $\mathbb{E}[\log K] \leq D(P||Q) + 2\log e$ and encoding $K$ using a universal code. As noted in~\cite{sfrl}, by instead using a power law code one can bound $\mathbb{E}[l(\mathcal{C}(K))] \leq D(P||Q) + \log (D(P||Q) + 1) + 7.78$ for discrete $P, Q$. The Poisson functional representation, first proposed by Li and El~Gamal in~\cite{sfrl}, has since been used to prove a tighter upper bound on $\mathbb{E}[l(\mathcal{C}(K))]$~\cite{li2024pointwise}. 
\begin{definition}[Poisson functional representation~\cite{li2021unified}] \label{def:pfr} Let $P$ and $Q$ be the target and proposal distributions, respectively, of an exact sampling problem. Let $\{T_i\}_{i \geq 1}$ be a rate-one Poisson process and let $\{U_i\}_{i \geq 1}$ be an i.i.d. sequence distributed according to $Q$ and independent of $\qty{T_i}_{i \geq 1}$. Define
\begin{equation}
    K \coloneqq \underset{i \geq 1}{\arg\min} \frac{T_i}{\frac{\dd P}{\dd Q}(U_i)}. \label{eq:K}
\end{equation}
\end{definition}
It is shown in~\cite[Appendix A]{li2021unified} that under this definition, $U_K \sim P$. For $P$ and $Q$ general distributions and $K$ generated according to Definition~\ref{def:pfr},~\cite{li2024pointwise} showed that
\begin{equation}
    \mathbb{E}[l(\mathcal{C}(K))] \leq D(P||Q) + \log(D(P||Q) + 2) + 3. \nonumber
\end{equation}
Moreover,~\cite{sfrl} showed that there exist distributions for which $\mathbb{E}[l(\mathcal{C}(K))] \geq D(P||Q) + \log(D(P||Q) + 1) - 1$, i.e., the $\log$ term in the upper bound cannot be removed in general. Recent improvements on these bounds are given in~\cite{goc2024channel, flamich2025redundancy}, where the channel simulation divergence $D_{CS}(P||Q)$ (see~\cite{goc2024channel} for the definition) is used to show
\begin{equation}
    D(P||Q) \leq D_{CS}(P||Q) \leq \mathbb{E}[l(\mathcal{C}(K))] \leq D_{CS}(P||Q) + \log(e + 1) + 1.
\end{equation}

\subsection{Contributions}
In this work, we study the Campbell cost, $L(t)$, of exact sampling algorithms. Although exact sampling and its associated minimum communication cost are well-studied problems, to the best of our knowledge, this natural extension (from linear cost to an exponential cost in message lengths) has not been investigated in the literature. Our main contributions can be summarized as follows:
\begin{itemize}
    \item We provide a lower bound on the $L(t)$ cost of encoding $K$ using any (possibly noncausal) sampling algorithm that grows approximately as $D_{\frac{1}{\alpha}}(P||Q)$, where $\alpha = \frac{1}{1 + t}$. We also prove the lower bound $L(t) \geq D_{2-\alpha}(P||Q) + \frac{1}{1-\alpha} \log(\frac{1}{2-\alpha})$, which is a tighter bound when $\alpha$ is close to $1$ or $D_\alpha(P||Q)$ is small.
    \item We use the Poisson functional representation to prove the upper bound $L(t) \leq (1 + \epsilon) D_{\frac{1 + \epsilon(1-\alpha)}{\alpha}}(P||Q) + c_1(\alpha, \epsilon)$, for any $0 < \alpha < 1$ and $\epsilon > 0$, where $c_1(\alpha, \epsilon)$ is a constant depending on $\alpha$ and $\epsilon$. Using numerical examples, we demonstrate that this upper bound is within 5-10 bits of the $D_{\frac{1}{\alpha}}(P||Q)$ lower bound, even for distributions where $L(t)$ is large. 
    \item We use our results to upper bound the R\'enyi entropy of channel simulation by the Sibson $\alpha$-mutual information~\cite{sibson1969information}, establishing a R\'enyi-generalized version of the strong functional representation lemma~\cite{sfrl}.
    \item We describe an operational procedure for encoding $K$ using a universal code that gives $L(t) \leq D_{\frac{2-\alpha}{\alpha}}(P||Q) + (1 + \epsilon) \log (D(P||Q) + 1) + c_2(\epsilon)$, for any $2/3 < \alpha < 1$ and $0 < \epsilon \leq \frac{3\alpha - 2}{2 - 2\alpha}$, where $c_2(\epsilon)$ is a constant. We show that this upper bound reduces to the bound of Harsha et al.~\cite{harsha2010communication} as $\alpha \to 1$.
    \item We characterize the minimum asymptotic Campbell cost of exact sampling between product distributions of dimension $n$, showing that $\lim_{n \to \infty} L^*_n(t) / n = D_{1/\alpha}(P||Q)$, where $\alpha = \frac{1}{1+t}$. 
    \item We show that under mild assumptions, a causal sampler has $\liminf_{n \to \infty} L^*_n / n \geq D_\beta(P||Q)$, where $\beta > \frac{1}{\alpha}$, and thus does strictly worse than a noncausal sampler in the asymptotic Campbell cost. 
\end{itemize}
Our main technical tool is to construct an encoding of $K$ that allows us to invoke the upper and lower bounds on the moments of $K$ in~\cite{liu2018rejection, li2021unified}. This requires, among other technical details, nontrivially extending the lower bound of~\cite[Thm. 7]{liu2018rejection} to any bijection of $K$ and the moment upper bound~\cite[Prop. 4]{li2021unified} beyond concave functions of $K$.
\section{Lower Bound on the Campbell Cost}
\label{sec:lowerbounds}
We present our main results as bounds on the Campbell cost $L(t)$; but, using Proposition~\ref{prop:campbellcost}, we can write them as bounds on $H_\alpha(K)$ within one bit (this will be formalized in Corollary~\ref{cor:renyientropybound}). In this section, we will prove a lower bound on the Campbell cost of any sampling algorithm. Note that for all $0 < \alpha < 1$ and with $t = \frac{1-\alpha}{\alpha}$, one has the trivial lower bound $L(t)\nolinebreak\geq\nolinebreak H_\alpha(K)\geq\nolinebreak H(K) \nolinebreak\geq\nolinebreak D(P||Q)$, where the final inequality follows similar to~\cite[Prop. IV.3]{harsha2010communication}. However, such a lower bound loses the dependency on $\alpha$, which is key to understanding $H_\alpha(K)$ and $L(t)$.
\begin{theorem} \label{thm:lowerbound}
    Let $K$ be the output of any exact sampling algorithm with proposal distribution $Q$, target distribution $P$, and common randomness $\qty{U_i}_{i \geq 1}$. Suppose $K$ is encoded using a uniquely decodable binary code $\mathcal{C}$. Then, for any $0 < \alpha < 1$ with $t = \frac{1-\alpha}{\alpha}$,
    \begin{equation}
       L(t) \geq \max \qty{\text{LB}_1^\alpha, \text{LB}_2^\alpha}, \label{eq:thmlowerbound}
    \end{equation}
    with $\text{LB}_1^\alpha \coloneqq D_{\frac{1}{\alpha}}(P||Q) + \frac{\alpha}{1-\alpha} \log(\alpha) - 1$ and $\text{LB}_2^\alpha \coloneqq D_{2-\alpha}(P||Q) + \frac{1}{1-\alpha}\log(\frac{1}{2-\alpha})$.
\end{theorem}
Theorem~\ref{thm:lowerbound} lower bounds the Campbell cost for any exact sampling problem. The constant term in $\text{LB}_1^\alpha$ ranges between $-1$ (at $\alpha = 0$) and $\frac{-1}{\ln 2} - 1$ (at $\alpha = 1$) while the constant term in $\text{LB}_2^\alpha$ ranges between $-1$ (at $\alpha = 0$) and $\frac{-1}{\ln 2}$ (at $\alpha = 1$). Thus, as $\alpha \rightarrow 1$ in~\eqref{eq:thmlowerbound} we recover the bound $\mathbb{E}[n_K] \geq D(P||Q) - \frac{1}{\ln(2)}$, the lower bound in~\cite{harsha2010communication} with a small penalty term of $\frac{-1}{\ln(2)}$. Except for $P$ and $Q$ close together (in the sense of R\'enyi divergence) or $\alpha$ close to 1, in~\eqref{eq:thmlowerbound}, the dominant term will typically be $D_{\frac{1}{\alpha}}(P||Q)$ and $\text{LB}_1^\alpha$ will yield the tighter bound. We include $\text{LB}_2^\alpha$ only to improve the lower bound for such situations. Numerical examples comparing $\text{LB}_1^\alpha$ and $\text{LB}_2^\alpha$, shown in Section~\ref{sec:simulations}, support this conclusion. 
Before proving Theorem~\ref{thm:lowerbound}, we state a lemma to be used in the proof.
\begin{lemma} \label{lemma:alphamomentlowerbound}
    Let $K$ be the output of any sampling algorithm with proposal distribution $Q$, target distribution $P$, and common randomness $\qty{U_i}_{i \geq 1}$. Then, for any $0 < \alpha < 1$ and bijection $g : \mathbb{N} \to \mathbb{N}$, 
    \begin{equation}
        \mathbb{E}[(g(K))^\alpha] \geq  \frac{1}{1 + \alpha} 2^{\alpha D_{\alpha + 1} (P||Q)}. \label{eq:alphamomentglowerbound}
    \end{equation}
\end{lemma}
\begin{proof}
    The lower bound~\eqref{eq:alphamomentglowerbound} was originally proved in~\cite{liu2018rejection} for countable $\mathcal{U}$ and $g(k) = k$, we extend the argument to Polish spaces and all bijections $g : \mathbb{N} \to \mathbb{N}$. The proof can be found in Appendix~\ref{app:kmomentproof}.
\end{proof}
\begin{proof}[Proof of Theorem~\ref{thm:lowerbound}]
    We will separately show that $L(t) \geq \text{LB}_1^\alpha$ and $L(t) \geq \text{LB}_2^\alpha$. We first prove that $\text{LB}_2^\alpha \leq H_\alpha(K)$; the desired inequality ($\text{LB}_2^\alpha \leq L(t)$) is immediate by the bound $H_\alpha(K) \leq L(t)$ from Proposition~\ref{prop:campbellcost}. Define $g : \mathbb{N} \rightarrow \mathbb{N}$ to be a bijection which maps $k$ to its index in $\qty{\mathbb{P}(K=k)}_{k \geq 1}$ sorted in nonincreasing order (with ties broken arbitrarily), so that with $\tilde{K} \coloneqq g(K)$ and $\tilde{p}(k) \coloneqq \mathbb{P}(\tilde{K} = k)$ one has $\tilde{p}(k) \geq \tilde{p}(k+1)$ for all $k \in \mathbb{N}$. Note that $H_\alpha(K) = H_\alpha(\tilde{K})$. For any $n \in \mathbb{N}$, 
    \begin{equation}
        1 = \sum_{k=1}^\infty \tilde{p}(k) \geq \sum_{k=1}^n \tilde{p}(k) \geq \sum_{k=1}^n \tilde{p}(n) = n\tilde{p}(n), \label{eq:knonincreasing}
    \end{equation}
    implying $\tilde{p}(n) \leq \frac{1}{n}$ for all $n$. Then, because $1-\alpha > 0$,
    \begin{align}
        \sum_{k=1}^\infty \tilde{p}(k)^\alpha &= \sum_{k=1}^\infty \tilde{p}(k) \left( \frac{1}{\tilde{p}(k)} \right)^{1-\alpha} \nonumber \\
        &\geq \sum_{k=1}^\infty \tilde{p}(k) k^{1-\alpha} = \mathbb{E}[\tilde{K}^{1-\alpha}]. \label{eq:1minusalphamoment}
    \end{align}
    Taking logarithms and dividing by $1 - \alpha$ gives
    \begin{align}
        H_\alpha(K) = H_\alpha(\tilde{K}) &\geq \frac{1}{1-\alpha}\log(\mathbb{E}[\tilde{K}^{1-\alpha}]) \nonumber  \\
        &\geq \frac{1}{1-\alpha} \log (\frac{1}{2-\alpha} 2^{(1-\alpha)D_{2-\alpha}(P||Q)} ) \label{eq:usedmomentlower} \\
        &= D_{2-\alpha}(P||Q) + \frac{1}{1-\alpha}\log(\frac{1}{2-\alpha}), \nonumber 
    \end{align}
    where~\eqref{eq:usedmomentlower} follows from Lemma~\ref{lemma:alphamomentlowerbound}. 
    
    We now prove that $L(t) \geq \text{LB}_1^\alpha$ by considering an encoding of $K$ which is one-to-one and not necessarily uniquely decodable. Here, a one-to-one code is an injective function $f : \mathbb{N} \to \qty{0,1}^*$. Consider a bijection $g:\mathbb{N} \rightarrow \mathbb{N}$ such that $\tilde{K} \coloneqq g(K)$ has probability distribution $\tilde{p}(k) \coloneqq \mathbb{P}(\tilde{K} = k)$ which is nonincreasing in $k$. The minimum Campbell costs of encoding $K$ and $\tilde{K}$ are the same, and the same is true for any one-to-one encoding of $\tilde{K}$. Therefore, because uniquely decodable codes are contained in the class of all one-to-one codes, showing that $L(t) \geq D_{\frac{1}{\alpha}}(P||Q) + \frac{\alpha}{1-\alpha} \log(\alpha) - 1 = \text{LB}_1^\alpha$ for any one-to-one encoding of $\tilde{K}$ will prove the theorem. Since the probability distribution of $\tilde{K}$ is nonincreasing in $k$; therefore, by the rearrangement inequality (see, e.g.,~\cite{steele2004cauchy}) the optimal one-to-one code $f^*$ under the Campbell cost must have nondecreasing codeword lengths, i.e., $l(f^*(k)) \leq l(f^*(k+1))$ for all $k$. Therefore, 
    an $f^*$ that maps $k$ to the $k$th element of $\qty{0,1}^*$, i.e., $f^*(1) = 0$, $f^*(2) = 1$, etc., is optimal. This optimal $f^*$ has codeword lengths $l(f^*(k)) = \lfloor\log(k+1) \rfloor$, and hence
\begin{align}
    L(t) = \frac{1}{t} \log(\mathbb{E}\left[ 2^{t l(f(\tilde{K}))} \right] ) &\geq \frac{1}{t} \log( \mathbb{E}\left[ 2^{tl(f^*(\tilde{K}))}\right]) \nonumber \\
     &= \frac{1}{t} \log (\sum_{k=1}^\infty \tilde{p}(k) 2^{t \lfloor \log (k+1) \rfloor} ) \nonumber \\
     &\geq \frac{1}{t} \log (\sum_{k=1}^\infty \tilde{p}(k) 2^{t(\log k -1)} ) \nonumber \\
     &= \frac{1}{t} \log (\mathbb{E}[\tilde{K}^t]) - 1 \nonumber \\
     &\geq \frac{1}{t}\log(\frac{1}{1 + t} 2^{tD_{t + 1}(P||Q)}) - 1 \label{eq:tmomentk} \\
     &= D_{\frac{1}{\alpha}}(P||Q) + \frac{\alpha}{1-\alpha}\log(\alpha) - 1, \nonumber
 \end{align}
 where~\eqref{eq:tmomentk} follows from Lemma~\ref{lemma:alphamomentlowerbound}.
\end{proof}
\section{Upper Bounds using the Poisson Functional Representation} 
\label{sec:upperbounds}
Theorem~\ref{thm:lowerbound} tells us that the minimum Campbell cost of any sampling algorithm grows at least as $D_{\frac{1}{\alpha}}(P||Q)$. The natural question is whether any one-shot sampling algorithms exist that achieve this lower bound. Suppose that we could bound $L(t) \leq D_{\frac{1}{\alpha}}(P||Q) + C$ for some constant $C$ and all distributions $P, Q$. Then, taking $\alpha \to 1$ (equivalently $t \to 0$) would yield $\mathbb{E}[l(\mathcal{C}(K))] \leq D(P||Q) + C$, a bound shown to be invalid by an explicit counter-example~\cite{sfrl}. Therefore, it is reasonable to conjecture that a tight upper bound in the form $L(t) \leq D_{\frac{1}{\alpha}}(P||Q) + c$, which is valid for all $P$ and $Q$, does not exist for the Campbell cost, either. Instead, we have the following upper bound on $L(t)$ using the Poisson functional representation.
\begin{theorem} \label{thm:upperbound}
    Let $K$ be the output of the Poisson functional representation given proposal distribution $Q$ and target distribution $P$. Then, for any $0 < \alpha < 1$ and $\epsilon > 0$, there exists a uniquely decodable encoding of $K$ such that
    \begin{equation}
        L(t) \leq (1 + \epsilon) D_{\frac{1 + \epsilon(1-\alpha)}{\alpha}}(P||Q) + c_1(\alpha, \epsilon), \label{eq:thmupperbound}
    \end{equation}
    with $t = \frac{1-\alpha}{\alpha}$ and $c_1(\alpha, \epsilon)$ a constant term defined as
    \begin{align}
        c_1(\alpha, \epsilon) &\coloneqq \begin{cases}
            (1 + \epsilon)\log e + 1 + \log(1 + \frac{1}{\epsilon}), &\frac{1}{2} < \alpha < 1 \text{ and } 0 < \epsilon < \frac{2\alpha - 1}{1-\alpha}\\
            \frac{\alpha}{1-\alpha} \log( \Gamma \left( \frac{1 + \epsilon(1-\alpha)}{\alpha} \right) ) + 4+3\epsilon - \frac{2\alpha}{1-\alpha}+ \log (1 + \frac{1}{\epsilon}), &0 < \alpha < \frac{1}{2} \text{ or } \epsilon \geq \frac{2\alpha - 1}{1-\alpha}.
        \end{cases} \label{eq:constanttermdef}
    \end{align}
\end{theorem}
We will refer to the upper bound in~\eqref{eq:thmupperbound} as $\text{UB}_1^\alpha \coloneq (1 + \epsilon) D_{\frac{1 + \epsilon(1-\alpha)}{\alpha}}(P||Q) + c_1(\alpha, \epsilon)$. We note that $c_1(\alpha, \epsilon)$ is finite for all $0 < \alpha < 1$, but tends to $\infty$ as $\alpha \to 0$. We can upper bound $c_1(\alpha, \epsilon)$ by a simpler function which more clearly shows the dependence on $\alpha$ and $\epsilon$.
\begin{proposition} \label{prop:simpleconstant}
    Let $c_1(\alpha, \epsilon)$ be defined as in~\eqref{eq:constanttermdef}. Then, for any $\epsilon > 0$ and $0 < \alpha < 1$,
    \begin{equation}
        c_1(\alpha, \epsilon) < \left( \frac{1}{2(1-\alpha)} + \frac{1}{2} + \epsilon \right) \log(\frac{1}{\alpha}) + \log(\frac{1}{\epsilon}) + 1.5\epsilon^2 + 4.5 \epsilon + 2.6 \eqqcolon \tilde{c}_1(\alpha, \epsilon).
    \end{equation}
\end{proposition}
Fig.~\ref{fig:constantcomparison} shows the value of $c_1(\alpha, \epsilon)$ and the looser bound in Proposition 2, minimized over $\epsilon > 0$ for each $0.2 < \alpha < 1$. Numerically, we see that $\displaystyle{\min_{\epsilon > 0} \tilde{c}_1(\alpha, \epsilon)}$ of Proposition~\ref{prop:simpleconstant} is approximately 2-3 bits looser than $\displaystyle{\min_{\epsilon > 0} c_1(\alpha, \epsilon)}$.
\begin{figure*}[ht]
    \centering
    \includegraphics{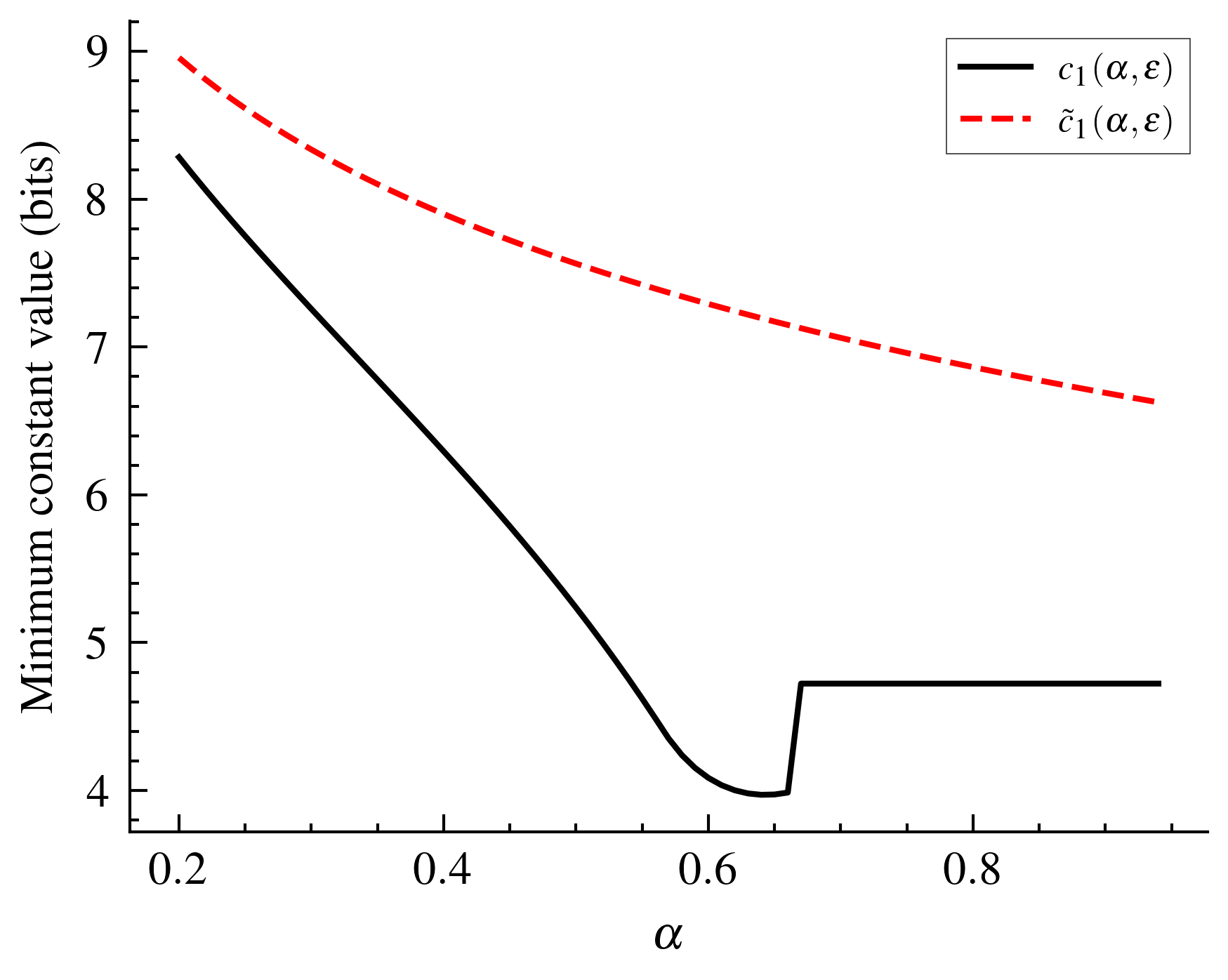}
    \caption{$\displaystyle{\min_{\epsilon > 0}} \; c_1(\alpha, \epsilon)$ and $\displaystyle{\min_{\epsilon > 0} \tilde{c}_1(\alpha, \epsilon)}$, as $\alpha$ varies between 0.2 and 1.}
    \label{fig:constantcomparison}
\end{figure*}

Before proving Theorem~\ref{thm:upperbound}, we state two lemmas which will be instrumental in upper-bounding the moments of $K$.
\begin{lemma} 
\label{lemma:kbounds}
    Let $K$ be the output of the Poisson functional representation given proposal distribution $Q$ and target distribution $P$. Then, 
    \begin{equation}
        \mathbb{E}\left[\log K \right] \leq D(P||Q) + 1, \label{eq:logbound}
    \end{equation}
    and for $0 < \alpha < 1$, 
    \begin{equation}
        \mathbb{E}\left[K^\alpha \right] \leq 2^{\alpha D_{\alpha + 1}(P||Q)} + \alpha. \label{eq:alphamomentbound}
    \end{equation}
\end{lemma}
Equation~\eqref{eq:logbound} was proved in~\cite{li2024pointwise}, and the upper bound in~\eqref{eq:alphamomentbound} follows from~\cite[Prop. 4]{li2021unified} after substituting $j=1$. We note that for $0 < \alpha < 1$ and $g(k) = k$, the Poisson functional representation in~\eqref{eq:alphamomentbound} almost achieves the lower bound on the $\alpha$-moment of any sampling algorithm in Lemma~\ref{lemma:alphamomentlowerbound}.
\begin{lemma}
    \label{lemma:geometricmoment}
    Let $X \sim \text{Geo}(p)$ be a geometric random variable with parameter $0 < p < 1$. Then, for any $r \geq 1$, 
    \begin{equation*}
        \mathbb{E}[X^r] \leq 2^{r-1} \left( \frac{\Gamma(r + 1)}{p^r} + 1 \right).
    \end{equation*}
\end{lemma}
\begin{proof}
    The proof can be found in Appendix~\ref{app:lemmageometricproof}.
\end{proof}
\begin{proof}[Proof of Theorem~\ref{thm:upperbound}]
We will prove Theorem~\ref{thm:upperbound} in two cases. In both cases, we will relate the Campbell cost $L(t)$ to the $rth$ moment of $K$, with $r = \frac{1-\alpha}{\alpha} (1 + \epsilon)$. For $r < 1$ (equivalently $\epsilon < \frac{2\alpha - 1}{1-\alpha}$), we will bound the moments of $K$ using~\eqref{eq:alphamomentbound} in Lemma~\ref{lemma:kbounds}, while in the case $r \geq 1$ (equivalently $\epsilon \geq \frac{2\alpha - 1}{1-\alpha}$) we will use the bound on the moments of a geometric random variable in Lemma~\ref{lemma:geometricmoment}. In particular, we will show that for any $1/2 < \alpha < 1$ and $0 < \epsilon < \frac{2\alpha - 1}{1-\alpha}$, with $t = \frac{1-\alpha}{\alpha}$,
\begin{equation}
    L(t) \leq (1 + \epsilon) D_{\frac{1 + \epsilon(1-\alpha)}{\alpha}}(P||Q) + (1 + \epsilon) \log e + 1 + \log (1 + \frac{1}{\epsilon}). \label{eq:upperbound1}
\end{equation}
Then, for any $0 < \alpha < 1$, for $t = \frac{1-\alpha}{\alpha}$ and $\epsilon \geq \frac{2\alpha - 1}{1-\alpha}$ (so that $t(1+\epsilon) \geq 1$) we will show that 
\begin{equation}
    L(t) \leq (1 + \epsilon) D_{\frac{1 + \epsilon(1-\alpha)}{\alpha}}(P||Q) + \frac{\alpha}{1-\alpha} \log ( \Gamma \left( \frac{1 + \epsilon(1-\alpha)}{\alpha} \right) ) + 4+3\epsilon - \frac{2\alpha}{1-\alpha} + \log (1 + \frac{1}{\epsilon}). \label{eq:upperbound2}
\end{equation}
Taking~\eqref{eq:upperbound1} and~\eqref{eq:upperbound2} together yields the statement of Theorem~\ref{thm:upperbound}. 

Fix $1/2 < \alpha < 1$ and $0 < \epsilon < \frac{2\alpha-1}{1-\alpha}$. Then, by the Kraft inequality there exists a uniquely decodable encoding of $K$ with codeword lengths $n_k \leq (1 +  \epsilon)\log k + 1 + \log(1 + \frac{1}{\epsilon})$; this is shown in Appendix~\ref{app:kraftanalysis}. Let $t = \frac{1-\alpha}{\alpha}$. Then, $0 < t < 1$ and $0 < \epsilon < \frac{2\alpha - 1}{1-\alpha} = \frac{1-t}{t}$ imply that $t(1+\epsilon)< 1$, so we can apply the upper bound on $\mathbb{E}\left[K^\alpha \right]$ in~\eqref{eq:alphamomentbound} to get
\begin{align}
L(t) &= \frac{1}{t}\log(\sum_{k=1}^\infty p(k) 2^{tn_k}) \nonumber \\
&\leq \frac{1}{t}\log(\sum_{k=1}^\infty p(k) 2^{t (1 + \epsilon) \log k + t  + t\log  ( 1 + \frac{1}{\epsilon})}) \nonumber \\
&= \frac{1}{t}\log(\mathbb{E}[K^{(1+\epsilon)t}]) + 1 + \log(1 + \frac{1}{\epsilon}) \nonumber \\
&\leq \frac{1}{t}\log(2^{(1 + \epsilon)tD_{(1 + \epsilon)t + 1}(P||Q)} + (1 + \epsilon)t) + 1 + \log(1 + \frac{1}{\epsilon}) \label{eq:kalphaboundused} \\
&\leq (1 + \epsilon)D_{(1+\epsilon)t + 1}(P||Q) + (1 + \epsilon)\log e + 1 + \log(1 + \frac{1}{\epsilon}). \label{eq:logsplit3}
\end{align}
Here,~\eqref{eq:kalphaboundused} follows from~\eqref{eq:alphamomentbound} and~\eqref{eq:logsplit3} from $\log(x + a) \leq \log(x) + a\log e$ for all $x \geq 1$ and $a > 0$. Noting that $(1 + \epsilon)t + 1 = \frac{1 + \epsilon(1-\alpha)}{\alpha}$ gives~\eqref{eq:upperbound1}. 

Suppose now that $0 < \alpha < 1$ and $ \epsilon \geq \frac{2\alpha - 1}{1-\alpha}$, so that with $t = \frac{1-\alpha}{\alpha}$, $(1 + \epsilon) t \geq 1$. From~\cite[Eq. 29]{li2021unified} (after substituting $j=1$), for any $u \in \mathcal{U}$, 
\begin{equation}
    K \mid \qty{U_K = u} \sim \text{Geo}(\beta(u)), \label{eq:distributionofK}
\end{equation}
for $\beta(u) = \mathbb{E}_{U \sim Q}\left[ \max \qty{\frac{\dd P}{\dd Q}(u), \frac{\dd P}{\dd Q}(U)} \right]^{-1}$. By Lemma~\ref{lemma:geometricmoment}, for $r = (1 + \epsilon)t \geq 1$, 
\begin{align}
    \mathbb{E}[K^r \mid U_K = u] &\leq 2^{r-1} \left( \frac{\Gamma(r + 1)}{\beta(u)^r} + 1 \right) \nonumber \\
    &= 2^{r-1} \Gamma(r + 1) \mathbb{E}_{U \sim Q} \left[\max \qty{\frac{\dd P}{\dd Q}(u), \frac{\dd P}{\dd Q}(U)} \right]^r + 2^{r-1} \nonumber \\
    &\leq 2^{r-1} \Gamma(r + 1) \left( \frac{\dd P}{\dd Q}(u) + \mathbb{E}_{U \sim Q}\left[ \frac{\dd P}{\dd Q} (U) \right] \right)^r + 2^{r-1} \nonumber \\
    &= 2^{r-1} \Gamma(r + 1) \left( \frac{\dd P}{\dd Q}(u) + 1 \right)^r + 2^{r-1} \nonumber \\
    &\leq \Gamma(r + 1) 2^{2r - 2} \left( \left( \frac{\dd P}{\dd Q}(u) \right)^r + 1 \right) + 2^{r-1}, \label{eq:boundagain}
\end{align}
where~\eqref{eq:boundagain} follows again by $(x + 1)^r \leq 2^{r-1} (x^r + 1)$ for any $r \geq 1$ and $x > 0$. Taking expectation with respect to $U_K \sim P$, we get that
\begin{equation}
    \mathbb{E}[K^r] \leq \Gamma(r + 1) 2^{2r-2} \left( 2^{r D_{r + 1}(P||Q)} + 1 \right) + 2^{r-1}. \label{eq:lemma4result}
\end{equation}
Again encoding $K$ using a uniquely decodable code with codeword lengths $n_k \leq (1 + \epsilon) \log k + 1 + \log (1 + \frac{1}{\epsilon})$ and writing $r = (1 + \epsilon)t$ and $c(\epsilon) = 1 + \log ( 1 + \frac{1}{\epsilon})$, we have that
\begin{align}
    L(t) &\leq \frac{1}{t} \log(\mathbb{E}[K^r]) + c(\epsilon) \nonumber \\
    &\leq \frac{1}{t} \log( \Gamma(r + 1) 2^{2r-2} \left( 2^{r D_{r + 1}(P||Q)} + 1 \right) + 2^{r-1}) + c(\epsilon) \label{eq:lemma4upperbound} \\
    &\leq \frac{1}{t} \log(\Gamma(r + 1) 2^{2r-2} \left( 2^{r D_{r + 1}(P||Q)} + 1 \right)) + \frac{1}{t} \log(2^{r-1}) + c(\epsilon) \label{eq:login1} \\
    &= \frac{1}{t}\log(2^{r D_{r + 1}(P||Q)} + 1) + \frac{1}{t} \log( \Gamma(r + 1)) + \frac{2r-2}{t} + \frac{r-1}{t} + c(\epsilon)\nonumber  \\
    &\leq \frac{1}{t}\log(2^{r D_{r + 1}(P||Q)}) + \frac{1}{t} \log( \Gamma(r + 1)) + \frac{2r-2}{t} + \frac{1}{t} + \frac{r-1}{t} + c(\epsilon) \label{eq:login2} \\
    &= \frac{r}{t} D_{r + 1}(P||Q) +\frac{1}{t} \log( \Gamma(r + 1)) + \frac{2r-2}{t} +  \frac{1}{t} + \frac{r-1}{t} + c(\epsilon). \nonumber
\end{align}
Here,~\eqref{eq:lemma4upperbound} follows by~\eqref{eq:lemma4result}, while~\eqref{eq:login1} and~\eqref{eq:login2} both follow from the inequality $\log(x + 1) \leq \log (x) + 1$ for all $x \geq 1$. Substituting $r = (1+\epsilon)t$ and $t = \frac{1-\alpha}{\alpha}$ gives~\eqref{eq:upperbound2} and the statement of the theorem. 
\end{proof}
The R\'enyi divergence term in $\text{UB}_1^\alpha$ has order $\frac{1 + \epsilon(1 -\alpha)}{\alpha}$, which as $\epsilon \rightarrow 0$ goes to $\frac{1}{\alpha}$, the order of the R\'enyi divergence term in the (generally) tighter $\text{LB}_1^\alpha$. However, $c_1(\alpha, \epsilon) \rightarrow \infty$ as $\epsilon \rightarrow 0$, thus $\text{UB}_1^\alpha$ is not tight. Numerical examples in Section~\ref{sec:simulations} demonstrate that, after minimizing over $\epsilon > 0$, $\text{UB}_1^\alpha$ is typically within 5-10 bits of $\text{LB}_1^\alpha$. As $\alpha \rightarrow 1$, we obtain the upper bound on the expected message length $\mathbb{E}[l(\mathcal{C}(K))] \leq (1 + \epsilon) D(P||Q) + (1+\epsilon)\log e + 1 + \log(1 + \frac{1}{\epsilon})$, for any $\epsilon > 0$. This bound is generally looser than the bound from the strong functional representation lemma~\cite{sfrl}, but better than the upper bound $\mathbb{E}[l(\mathcal{C}(K))] \leq D(P||Q) + (1 + \epsilon)\log(D(P||Q) + 1) + c_\epsilon$ by Harsha et al.~\cite{harsha2010communication}. Encoding $K$ using the Elias omega code~\cite{elias2003universal}, as done in~\cite{harsha2010communication}, yields the following generalization of their result, which is strictly worse than $\text{UB}_1^\alpha$.
\begin{theorem} \label{thm:worseupperbound}
    Let $K$ be the output of the Poisson functional representation given proposal distribution $Q$ and target distribution~$P$. Then, for any $2/3 < \alpha < 1$ and $0 < \epsilon \leq \frac{3\alpha - 2}{2-2\alpha}$, encoding $K$ using a universal code with codeword lengths $n_k = \log k + (1+\epsilon)\log \log (k + 1) + O(1)$ gives
    \begin{equation}
        L(t) \leq D_{\frac{2-\alpha}{\alpha}}(P||Q) + (1 + \epsilon)\log(D(P||Q) + 1) + c_2(\epsilon), \label{eq:worseupperbound}
    \end{equation}
    with $t = \frac{1-\alpha}{\alpha}$ and $c_2(\epsilon) = 3 + \epsilon + \log(\frac{\ln(2)}{\epsilon} + \frac{3}{2})$. 
\end{theorem}
\begin{proof}
    Let $2/3 < \alpha < 1$ and $0 < \epsilon \leq \frac{3\alpha - 2}{2-2\alpha}$. We will bound $L(t)$ using a universal coding of the natural numbers. In particular, using the prefix-free encoding described in~\cite[Ex. 1.11.16]{li2008introduction} (see also~\cite{harsha2010communication}) we can set $n_k \leq \log k + (1 + \epsilon) \log \log (k + 1) + 1 + \log (\frac{\ln(2)}{\epsilon} + \frac{3}{2})$ for any $\epsilon > 0$; this is proved  in Appendix~\ref{app:kraftanalysis}. Since $t = \frac{1-\alpha}{\alpha}$, $0 < t < 1/2$. Then,
    \begin{align}
        L(t) &= \frac{1}{t}\log(\mathbb{E}\left[ 2^{tn_K} \right]) \nonumber  \\
        &\leq \frac{1}{t} \log( \mathbb{E} \left[ K^t \log(K+1)^{(1+\epsilon)t} \right]) \nonumber + 1 + \log (\frac{\ln (2)}{\epsilon} + \frac{3}{2}) \\
        &\leq \frac{1}{t}\log(\sqrt{\mathbb{E}[K^{2t}]\mathbb{E}[\log(K+1)^{2(1+\epsilon)t}]}) + 1 + \log (\frac{\ln (2)}{\epsilon} + \frac{3}{2}) \label{eq:cauchyschwartz} \\
        &= \frac{1}{2t}\left( \log(\mathbb{E}[K^{2t}]) + \log(\mathbb{E}[\log(K+1)^{2(1+\epsilon)t}]) \right) + 1 + \log (\frac{\ln (2)}{\epsilon} + \frac{3}{2}). \label{eq:intermediateupperbound}
    \end{align}
    Here,~\eqref{eq:cauchyschwartz} follows from the Cauchy-Schwartz inequality. We will bound each of the two terms in~\eqref{eq:intermediateupperbound} separately, beginning with $\mathbb{E}[\log(K+1)^{2(1+\epsilon)t}]$. As $0 < t < \frac{1}{2}$, choosing $0 < \epsilon \leq \frac{1-2t}{2t} = \frac{3\alpha-2}{2-2\alpha}$ gives that $2(1+\epsilon)t \leq 1$, and so Jensen's inequality can be applied as
    \begin{align}
        \frac{1}{2t}\log \left( \mathbb{E}[\log(K + 1)^{2(1+\epsilon)t}] \right) &\leq \frac{1}{2t} \log \left( \mathbb{E}[\log (K + 1)]^{2(1+\epsilon)t} \right) \nonumber \\
    &= (1 + \epsilon) \log \big( \mathbb{E}[\log (K + 1)] \big) \nonumber \\
    &\leq (1 + \epsilon) \log \big ( \mathbb{E}[\log K] + 1 \big) \label{eq:logsplit} \\
    &\leq (1 + \epsilon) \log(D(P||Q) + 1 + 1) \label{eq:logkintermediateupperbound} \\
    &\leq (1 + \epsilon) \log(D(P||Q) + 1) + 1 + \epsilon. \label{eq:logsplit2}
    \end{align}
    Here,~\eqref{eq:logsplit} and~\eqref{eq:logsplit2} both follow from the inequality $\log(x+1) \leq \log(x) + 1$ for all $x \geq 1$ and~\eqref{eq:logkintermediateupperbound} follows from the upper bound~\eqref{eq:logbound} on $\mathbb{E}\left[\log K\right]$ in Lemma~\ref{lemma:kbounds}. In the first term of~\eqref{eq:intermediateupperbound}, because $0 < t < 1/2$ by assumption, we can apply the upper bound~\eqref{eq:alphamomentbound} on $\mathbb{E}\left[K^\alpha\right]$ in Lemma~\ref{lemma:kbounds} to get that
    \begin{align}
        \frac{1}{2t}\log(\mathbb{E}[K^{2t}]) &\leq \frac{1}{2t}\log( 2^{2tD_{2t+1}(P||Q)} + 2t) \nonumber \\
        &\leq D_{2t+1}(P||Q) + 1 \label{eq:kalphaintermediateupperbound},
    \end{align}
    where~\eqref{eq:kalphaintermediateupperbound} holds again by $\log(x + a) \leq \log(x) + a$. Combining~\eqref{eq:logsplit2} and~\eqref{eq:kalphaintermediateupperbound}, and noting $2t+1 = \frac{2-\alpha}{\alpha}$, we obtain
    \begin{align}
        L(t) \leq D_{\frac{2-\alpha}{\alpha}}(P||Q) + (1 + \epsilon)\log(D(P||Q) + 1) + c_2(\epsilon), \nonumber
    \end{align}
    with $c_2(\epsilon) = 2 + \epsilon + \log(\frac{\ln2}{\epsilon} + \frac{3}{2})$.
\end{proof}
We will refer to the upper bound in Theorem~\ref{thm:worseupperbound} as $\text{UB}_2^\alpha \coloneqq D_{\frac{2-\alpha}{\alpha}}(P||Q) + (1 + \epsilon)\log(D(P||Q) + 1) + c_2(\epsilon)$. As $\alpha \to 1$, $\text{UB}_1^\alpha$ reduces to $\mathbb{E}[l(\mathcal{C}(K))] \leq D(P||Q) + (1 + \epsilon) \log (D(P||Q) + 1) + c_2(\epsilon)$ for any $\epsilon > 0$, the same upper bound as in~\cite{harsha2010communication}. Note that $\text{UB}_2^\alpha$ is valid only for $2/3 < \alpha < 1$, whereas $\text{UB}_1^\alpha$ is valid for all $0 < \alpha < 1$. Moreover, for all $2/3 < \alpha < 1$, $\text{UB}_1^\alpha < \text{UB}_2^\alpha$. To see why, observe that for $\epsilon < 1$, $\frac{1 + \epsilon(1-\alpha)}{\alpha} < \frac{2-\alpha}{\alpha}$, i.e., the order of the leading R\'enyi divergence term is strictly less in $\text{UB}_1^\alpha$. Since $D_\alpha(P||Q)$ is nondecreasing in $\alpha$ and is strictly increasing unless $P = Q(\,\cdot\, | A)$ for some event $A$ (see, e.g.,~\cite[Thm. 3]{van2014renyi}) one has $D_{\frac{1 + \epsilon(1-\alpha)}{\alpha}}(P||Q) < D_{\frac{2-\alpha}{\alpha}}(P||Q)$ for general $P$ and $Q$. For such distributions, $(1 + \epsilon)D_{\frac{1 + \epsilon(1-\alpha)}{\alpha}}(P||Q) < D_{\frac{2-\alpha}{\alpha}}(P||Q)$ for $\epsilon > 0$ small enough. As $c_1(\alpha, \epsilon) < c_2(\epsilon)$ for all $2/3 < \alpha < 1$, $\text{UB}_1^\alpha$ is better even for small $D_\alpha(P||Q)$. Numerical results, shown in Section~\ref{sec:simulations}, illustrate this conclusion. 

While Theorems~\ref{thm:lowerbound} to~\ref{thm:worseupperbound} are bounds on the Campbell cost, we can use Proposition~\ref{prop:campbellcost} to obtain bounds on $H_\alpha(K)$. 
\begin{corollary} \label{cor:renyientropybound}
Let $K$ be the output of any exact sampling algorithm with proposal distribution $Q$, target distribution $P$, and common randomness $\qty{U_i}_{i \geq 1}$. For any $0 < \alpha < 1$, $H_\alpha(K) > \max \qty{\text{LB}_1^\alpha, \text{LB}_2^\alpha} - 1$, with $\text{LB}_1^\alpha$ and $\text{LB}_2^\alpha$ given in Theorem~\ref{thm:lowerbound}. Moreover, for $K$ the output of the Poisson functional representation, for any $0 < \alpha < 1$ and $\epsilon > 0$ we have that $H_\alpha(K) \leq \text{UB}_1^\alpha$, with $\text{UB}_1^\alpha$ given in Theorem~\ref{thm:upperbound}. 
\end{corollary}
\begin{proof}
    The upper bound is immediate after substitution in~\eqref{eq:prop1lower} in Proposition~\ref{prop:campbellcost}. The lower bound follows from~\eqref{eq:prop1upper} in Proposition~\ref{prop:campbellcost}, specifically that for any sampling algorithm outputting $K$, there exists an encoding of $K$ such that $L(t) < H_\alpha(K) + 1$. Then, the lower bound $L(t) \geq \max \qty{\text{LB}^\alpha_1, \text{LB}^\alpha_2}$ of Theorem~\ref{thm:lowerbound} applies to this sampling algorithm and encoding, meaning $H_\alpha(K) > \max \qty{\text{LB}^\alpha_1, \text{LB}^\alpha_2} - 1$. 
\end{proof}
In the problem of exact channel simulation used to motivate this paper, we have the following corollaries to Theorem~\ref{thm:upperbound}. 
\begin{corollary}
Let $X$ and $Y$ by arbitrary random variables with Polish alphabets $\mathcal{X}$ and $\mathcal{Y}$. The sender is provided with $x \in \mathcal{X}$ generated according to $\text{P}_X$ and wishes to transmit a message $K$ so that the receiver can generate $y \in \mathcal{Y}$ with distribution $\text{P}_{Y \mid X}( \, \cdot \, | x)$. We assume the sender and receiver have shared access to the sequence $\qty{U_i}_{i \geq 1}$ i.i.d. according to the proposal distribution $Q_{Y_\alpha}$ defined in~\eqref{eq:sibsondistribution}. Let $K$ be the output of the Poisson functional representation. Then, for any $0 < \alpha < 1$ and $\epsilon > 0$,
\begin{equation}
    H_\alpha(K) \leq (1 + \epsilon) I^{\text{\scriptsize s}}_{\frac{1 + \epsilon(1-\alpha)}{\alpha}}(X ;Y) + c_1(\alpha, \epsilon), \label{eq:channelsimulationupperbound}
\end{equation}
where $I_\alpha^{\text{\scriptsize s}} (X; Y)$ is the Sibson $\alpha$-mutual information defined in~\eqref{eq:sibsondefinition} and $c_1(\alpha, \epsilon)$ is as in Theorem~\ref{thm:upperbound}.
\label{cor:channelsimulation}
\end{corollary}
\begin{remark}
    The distribution $Q_{Y_\alpha}$ minimizes $D_\alpha(\text{P}_{XY} || \text{P}_X \times Q_Y)$ over all proposal distributions. Notably, the minimizer is not the marginal distribution $\text{P}_Y$. This differs from the Shannon setting, where selecting the proposal distribution to be $\text{P}_Y$ is known to be optimal. 
\end{remark}
\begin{remark}
    Corollary~\ref{cor:channelsimulation} can be interpreted as a R\'enyi-generalized version of the strong functional representation lemma~\cite{sfrl}. Indeed, for $Z = \qty{U_i}_{i \geq 1}$ the common randomness, we have that
    \begin{equation*}
        H_\alpha(Y \mid Z) \leq H_\alpha(K, Z \mid Z) = H_\alpha(K \mid Z) \leq H_\alpha(K). 
    \end{equation*}
    As such,~\eqref{eq:channelsimulationupperbound} upper bounds the R\'enyi conditional entropy $H_\alpha(Y \mid Z)$. Here, the definition of the R\'enyi conditional entropy can be any which satisfies both conditioning reduces entropy and the data processing inequality; examples include the Arimoto~\cite{arimoto1977information} and Hayashi-Skoric conditional entropies\cite{hayashi2011exponential, vskoric2011sharp, kamatsuka2025several}. 
\end{remark}
\begin{proof}   
    As in the proof of~\cite[Thm. 1]{sfrl}, we condition on $X = x$ and take expectation of the upper bound on $\mathbb{E}[K^r \mid X = x]$ in~\eqref{eq:alphamomentbound} and~\eqref{eq:lemma4result}. Choosing $P = \text{P}_{Y \mid X}(\cdot \mid x)$ and $Q = Q_{Y_\alpha}$, we obtain
    \begin{equation*}
        \mathbb{E}[K^r \mid X = x] \leq \gamma_1(r) 2^{r D_{r + 1}({\text{P}_{Y \mid X}(\cdot \mid x)} || Q_{Y_\alpha})} + \gamma_2(r),
    \end{equation*}
    where $r >0$ and $\gamma_1(r), \gamma_2(r)$ are constants as in~\eqref{eq:alphamomentbound} for $0 < r < 1$ and~\eqref{eq:lemma4result} for $r \geq 1$. Taking expectation with respect to $X$ gives
    \begin{equation*}
        \mathbb{E}[K^r] \leq \gamma_1(r) \mathbb{E}_X\left[ 2^{r D_{r + 1}({\text{P}_{Y \mid X}(\cdot \mid X)} || Q_{Y_\alpha})}\right]  + \gamma_2(r) = \gamma_1(r) 2^{r I_{r+1}^{\text{\scriptsize s}}(X; Y)}  + \gamma_2(r).
    \end{equation*}
    Encoding $K$ as in the proof of Theorem~\ref{thm:upperbound}, combined with this unconditional moment upper bound, gives~\eqref{eq:channelsimulationupperbound}. 
\end{proof}
\begin{corollary}
    In the case of channel simulation with arbitrary input, i.e., where the distribution $\text{P}_X$ is not known, for any $0 < \alpha < 1$ and $\epsilon > 0$, we have
    \begin{equation}
    \max_{x \in \mathcal{X}} H_\alpha(K \mid X = x) \leq (1 + \epsilon) C^{\text{\scriptsize s}}_{\frac{1 + \epsilon(1-\alpha)}{\alpha}} + c_1(\alpha, \epsilon), \label{eq:sibsonupperbound}
\end{equation}
where $C^{\text{\scriptsize s}}_\alpha \coloneqq \max_{\text{P}_X} I_\alpha^{\text{\scriptsize s}}(X; Y) = \min_{Q_Y} \max_{x \in \mathcal{X}} D_{\alpha} (\text{P}_{Y | X}( \cdot | x) || Q_Y)$ is the Sibson capacity of the channel $X \to Y$~\cite{csiszar2002generalized}.
\end{corollary}
\begin{proof}
    Fix $x \in \mathcal{X}$, $P = \text{P}_{Y \mid X}(\cdot \mid x)$, and $Q = Q_Y$. Theorem~\ref{thm:upperbound} gives that
    \begin{equation*}
        H_\alpha(K \mid X = x) \leq (1 + \epsilon) D_{\frac{1 + \epsilon(1-\alpha)}{\alpha}}(\text{P}_{Y \mid X}(\cdot \mid x) || Q_Y) + c_1(\alpha, \epsilon).
    \end{equation*}
    Maximizing over $x \in \mathcal{X}$ and minimizing over $Q_Y$ gives~\eqref{eq:sibsonupperbound}.
\end{proof}

As $\alpha \to 1$ in~\eqref{eq:channelsimulationupperbound}, the leading term becomes $I(X;Y)$, Shannon's mutual information. Similarly, the Sibson capacity in~\eqref{eq:sibsonupperbound} recovers the Shannon capacity as $\alpha \to 1$. 
\section{Asymptotic Results} \label{sec:asymptotics}
A natural extension of the one-shot exact sampling problem is the asymptotic problem, where the common randomness is the i.i.d. sequence $\qty{U_i}_{i \geq 1} \sim Q^{\otimes n}$ and we wish to sample from the product distribution $P^{\otimes n}$ by communicating a single message. In the setting of expected message length, letting $R^*_n$ be the minimum expected message length over all exact samplers with proposal $Q^{\otimes n}$ and target $P^{\otimes n}$, the optimal asymptotic communication rate (bits per sample) is known to be~\cite{harsha2010communication}
\begin{equation}
    \lim_{n \to \infty} \frac{R^*_n}{n} \coloneqq \lim_{n \to \infty} \frac{1}{n}\mathbb{E}[l(\mathcal{C}(K))] = D(P||Q). \nonumber
\end{equation}
Similarly, in the problem of i.i.d. channel simulation, where upon input $x^n$ drawn via $P_{X^n}$ one simulates a sample from the channel $\text{P}_{Y^n \mid X^n}(\; \cdot \, \mid x^n)$, the achievable rate is $\lim_{n \to \infty} R^*_n / n = I(X; Y)$. We can derive analogous asymptotic generalizations for the Campbell cost, $L(t)$. 
\begin{theorem}
    Given distributions $P \ll Q$, for any $t > 0$ let $L^*_n(t)$ be the minimum achievable Campbell cost among all exact sampling algorithms with target $P^{\otimes n}$ and common randomness $\qty{U_i}_{i \geq 1} \sim Q^{\otimes n}$. Then, with $\alpha = \frac{1}{t+1}$,  
    \begin{equation}
        \lim_{n \to \infty} \frac{L^*_n(t)}{n} = D_{\frac{1}{\alpha}}(P||Q). \label{eq:asymptoticcost}
    \end{equation}
\end{theorem}
\begin{proof}
For any $n$, $t > 0$, and sampling algorithm, Theorem~\ref{thm:lowerbound} yields the lower bound
    \begin{equation}
        \frac{L(t)}{n} \geq \frac{1}{n}D_\frac{1}{\alpha}(P^{\otimes n} || Q^{\otimes n}) - \frac{1}{n} \left( \frac{\alpha}{1-\alpha} \log \alpha - 1 \right). \nonumber
    \end{equation}
    The additivity of the R\'enyi divergence for product distributions~\cite{van2014renyi} gives $D_\alpha(P^{\otimes n} || Q^{\otimes n}) = nD_\alpha(P||Q)$, which implies the lower bound $\liminf_{n \to \infty} L^*_n(t) / n \geq D_{1/\alpha}(P||Q)$. The upper bound on $L(t)$ from Theorem~\ref{thm:upperbound} gives that
    \begin{equation}
        \frac{L(t)}{n} \leq \frac{(1 + \epsilon)}{n}D_{\frac{1 + \epsilon(1-\alpha)}{\alpha}}(P^{\otimes n} || Q^{\otimes n}) + \frac{c_1(\alpha, \epsilon)}{n}, \nonumber
    \end{equation}
    for any $n$, $t > 0$, and $\epsilon > 0$. Taking $n \to \infty$, and again applying additivity, we obtain 
    \begin{equation}
        \limsup_{n \to \infty} \frac{L^*_n(t)}{n} \leq (1 + \epsilon)D_{\frac{1 + \epsilon(1-\alpha)}{\alpha}}(P||Q). \label{eq:asymptoticupperboundepsilon}
    \end{equation}
As $\epsilon > 0$ in~\eqref{eq:asymptoticupperboundepsilon} is arbitrary, we can take $\epsilon \to 0$ to get that $\limsup_{n \to \infty} L^*_n(t) / n \leq D_{1 / \alpha}(P||Q)$. Here, we have used the continuity of the R\'enyi divergence in the order to conclude that $\lim_{\epsilon \to 0} D_\frac{1 -\epsilon(1-\alpha)}{\alpha}(P||Q) = D_{1/\alpha}(P||Q)$~\cite{van2014renyi}. The result follows. 
\end{proof}
As $\alpha \to 1$ in~\eqref{eq:asymptoticcost} we recover the optimal communication rate $\lim_{n \to \infty} R^*_n / n = D(P||Q)$. In~\cite{liu2018rejection}, Liu and Verd\'u showed that a causal sampler has a strictly worse expected $\alpha$-moment of $K$ than a noncausal sampler. In particular, they showed (under some regularity assumptions) that for a causal sampler and any $\alpha \in (0, 1)$, $\mathbb{E}[K^\alpha] \approx 2^{\alpha D_{\frac{1}{1-\alpha}}(P||Q)}$. In contrast, for a noncausal sampler, $\mathbb{E}[K^\alpha] \approx 2^{\alpha D_{\alpha + 1}(P||Q)}$. Under mild assumptions, we see a similar fundamental gap in the asymptotic Campbell cost between causal and noncausal sampling algorithms.
\begin{theorem}
    Given distributions $P \ll Q$, for any $t > 0$ let $L^*_n(t)$ be the minimum achievable Campbell cost among all causal sampling algorithms between proposal $Q^{\otimes n}$ and target $P^{\otimes n}$ such that the transmitted index always has a nonincreasing probability distribution. If $t < 1$, we assume that $D_{\frac{1}{1-t} + \delta}(P||Q)$ is finite for some $\delta > 0$. Then, with $\alpha = \frac{1}{t+1}$ we have that
    \begin{equation}
        \liminf_{n \to \infty}\frac{L^*_n(t)}{n} \geq D_{\beta}(P||Q),
    \end{equation}
    where 
    \begin{equation}
        \beta = \begin{cases}
            \frac{\alpha}{2\alpha - 1}, &\alpha \in (\frac{1}{2}, 1) \\
            \infty, &\alpha \in (0, \frac{1}{2}].
        \end{cases} \label{eq:beta}
    \end{equation}
    \label{thm:causallowerbound}
\end{theorem}
\begin{proof}
    The proof can be found in Appendix~\ref{app:causalasymptoticproof}. 
\end{proof}
The assumption in Theorem~\ref{thm:causallowerbound} that $K$ have a nonincreasing probability distribution is far less restrictive than it may first appear. All state-of-the-art exact causal sampling algorithms exhibit this property, including rejection sampling, greedy rejection sampling~\cite{harsha2010communication,flamich2023adaptive}, and greedy Poisson rejection sampling~\cite{flamich2023greedy}. 
\par
In the case of expected message length, causal and noncausal samplers achieve the same optimal communication rate of $D(P||Q)$. With $\beta$ as in~\eqref{eq:beta}, we have that $\beta > 1/\alpha$ and thus $D_\beta(P||Q) > D_{1/\alpha}(P||Q)$ in general. Interestingly, noncausal samplers have a strictly lower optimal asymptotic Campbell cost than causal samplers, and the gap is often significant. For example, if $P$ and $Q$ are Gaussians with different means, for $\alpha \leq 1/2$ a noncausal sampler can achieve an asymptotic cost of $D_{1/\alpha}(P||Q)$, while the asymptotic Campbell cost per sample of a causal sampler, such as greedy rejection sampling or greedy Poisson rejection sampling, is infinite.

Intuitively, a noncausal sampler outperforms a causal sampler by ``looking ahead" in the sequence of samples and resisting the temptation to select the current sample if a better one appears shortly afterwards. The Poisson functional representation appears to realize the advantage of noncausal samplers, and is near-optimal in the Campbell cost setting. Moreover, our results suggest the conjecture that the Poisson functional representation may even be optimal for exact sampling independent of the chosen cost.
\section{Numerical Examples} 
In this section, we give numerical examples comparing $\text{LB}_1^\alpha$, $\text{LB}_2^\alpha$, $\text{UB}_1^\alpha$, and $\text{UB}_2^\alpha$ from Theorems~\ref{thm:lowerbound} to~\ref{thm:worseupperbound}. Let $P$ and $Q$ be normal distributions on $\mathbb{R}$, specifically, for $\mu_1, \mu_2 \in \mathbb{R}$ and $\sigma_1, \sigma_2 \in \mathbb{R}_{\geq 0}$, let $P = \mathcal{N}(\mu_1, \sigma_1)$ and $Q = \mathcal{N}(\mu_2, \sigma_2)$. Trying to losslessly communicate a sample from a normal distribution is not possible; however, by allowing access to a shared source of randomness, we can communicate the sample with finite cost. Fig.~\ref{fig:bounds} compares the bounds for three different choices of the normal distribution parameters; note that the true minimum $L(t)$ is guaranteed to lie between $\text{UB}_1^\alpha$ and $\max \qty{\text{LB}_1^\alpha, \text{LB}_2^\alpha}$. We have only shown the bounds for $0.2 < \alpha < 1$ for display reasons, as both $\text{LB}_1^\alpha$ and $\text{UB}_1^\alpha$ go to $\infty$ as $\alpha \to 0$. 
\label{sec:simulations}
\begin{figure*}[ht]
    \centering
    \begin{subfigure}{0.32\textwidth}
        \includegraphics[width=\textwidth]{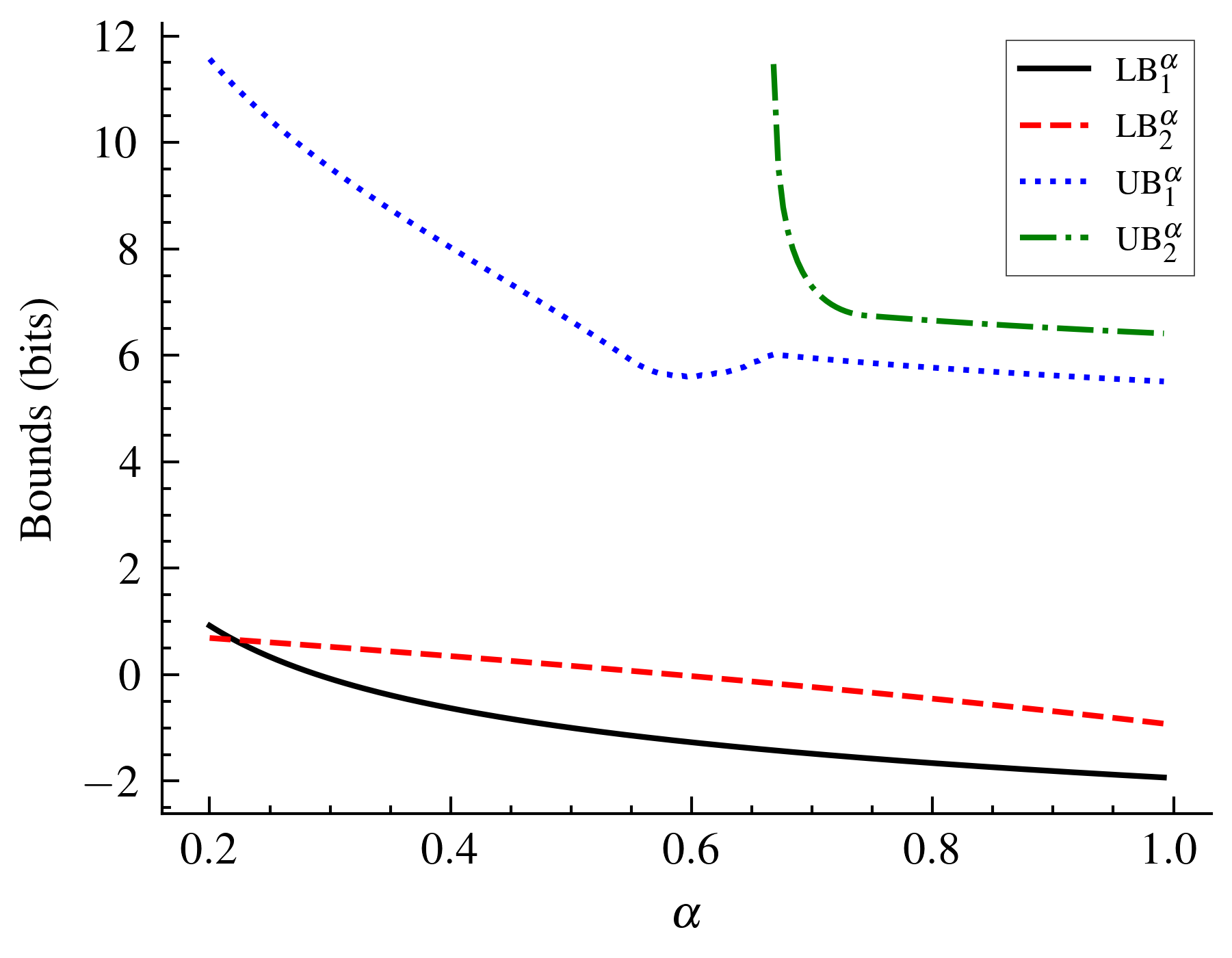}
        \caption{$P = \mathcal{N}(0, 1)$ and $Q = \mathcal{N}(1,1)$}
    \end{subfigure}
    \hfill
    \begin{subfigure}{0.32\textwidth}
        \includegraphics[width=\textwidth]{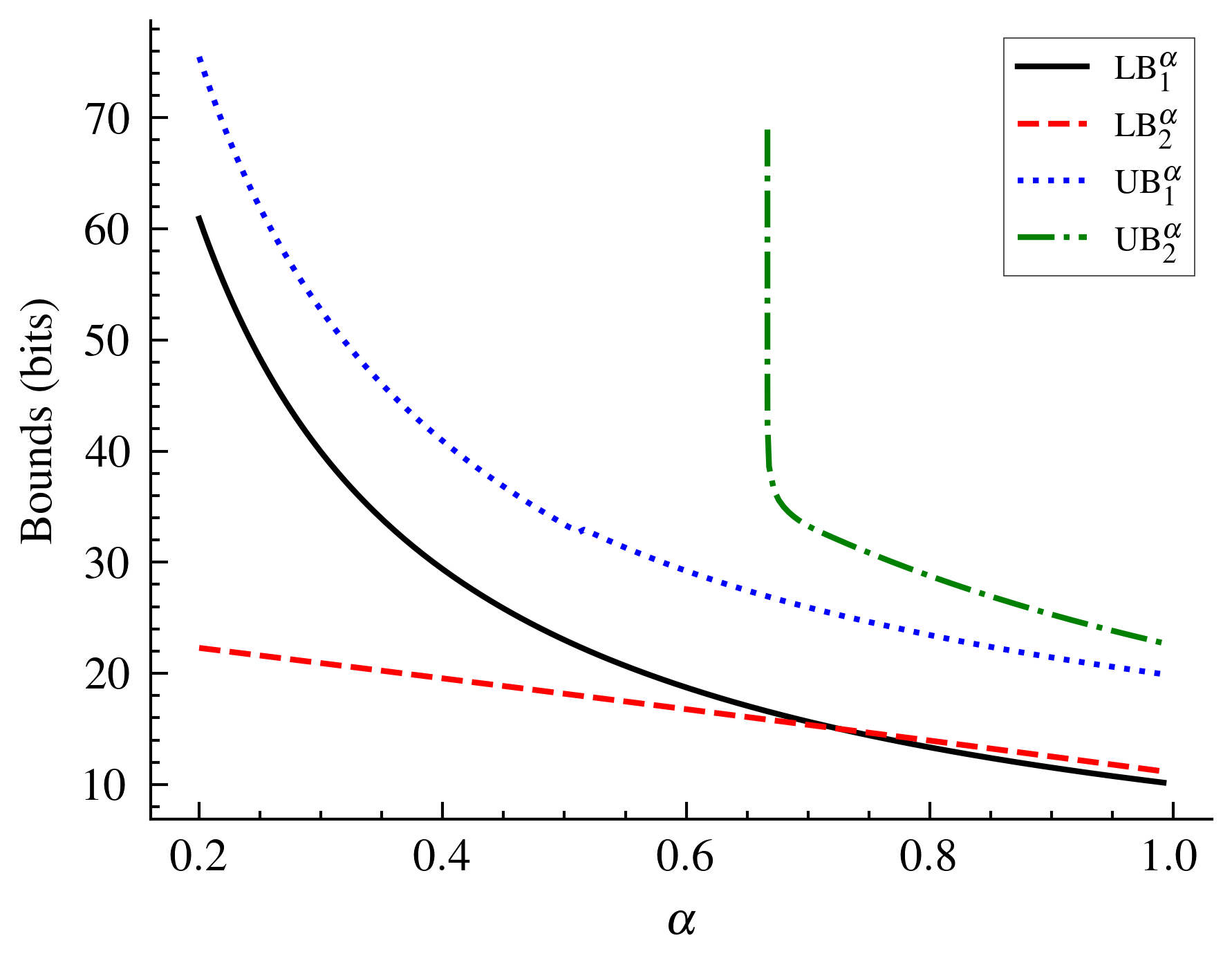}
        \caption{$P = \mathcal{N}(0, 1)$ and $Q = \mathcal{N}(5,1)$}
    \end{subfigure}
    \hfill
    \begin{subfigure}{0.32\textwidth}
        \includegraphics[width=\textwidth]{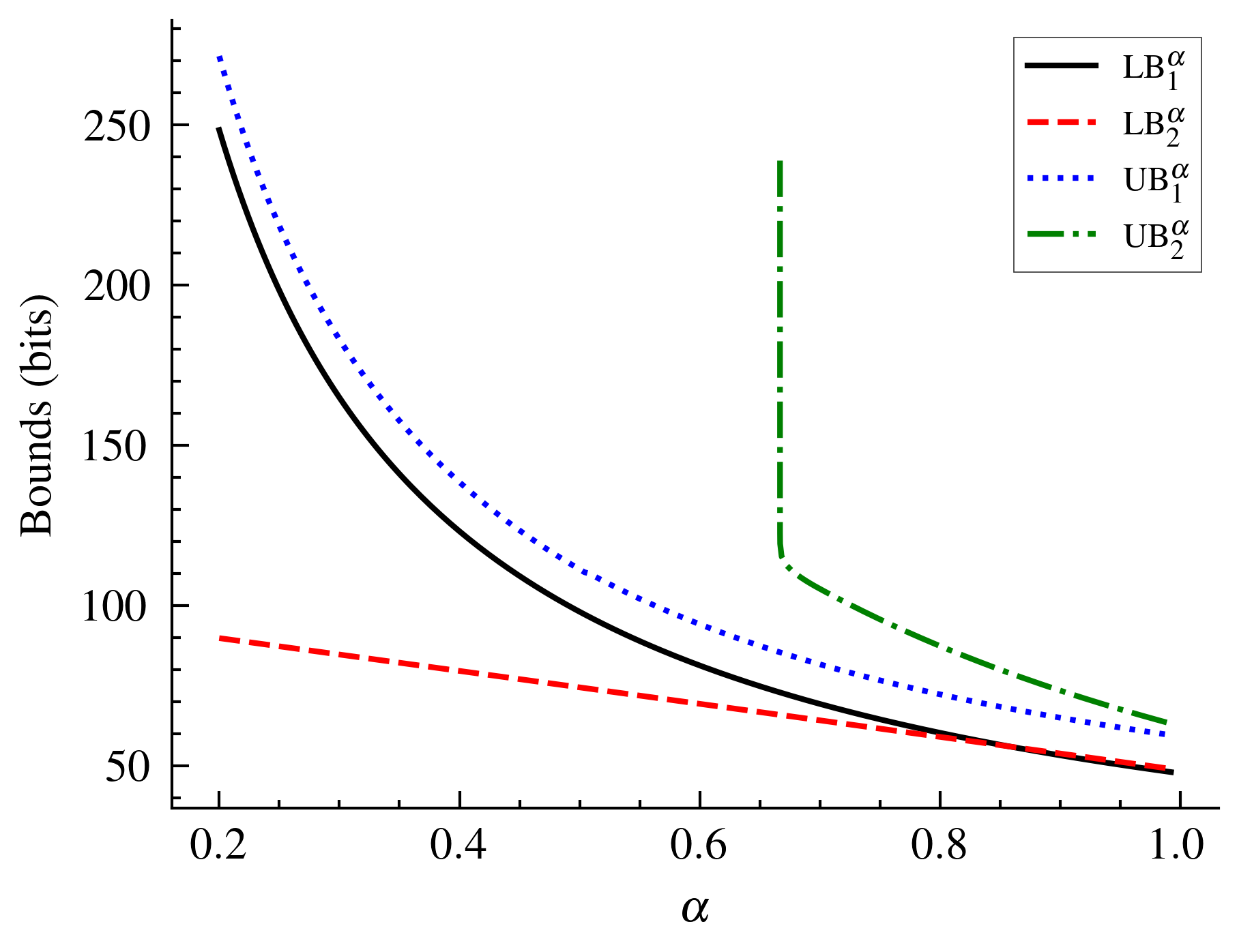}
        \caption{$P = \mathcal{N}(0, 1)$ and $Q = \mathcal{N}(10,1)$}
    \end{subfigure}
    \caption{Comparison of the bounds on the Campbell cost for two normal distributions, for $0.2 < \alpha < 1$.}
    \label{fig:bounds}
\end{figure*}
In both upper bounds, for each value of $\alpha$, the bound is optimized over $\epsilon$ to find the minimum value. As discussed in Section~\ref{sec:upperbounds}, in all cases $\text{UB}_1^\alpha$ is tighter than $\text{UB}_2^\alpha$ for all $2/3 < \alpha < 1$. Since $P = \mathcal{N}(\mu_1, \sigma_1)$ and $Q = \mathcal{N}(\mu_2, \sigma_2)$, we have that 
\begin{equation}
D_{\alpha}(P || Q) = \ln(\frac{\sigma_2}{\sigma_1}) + \frac{\ln(\frac{\sigma_2^2}{(\sigma^2)^*_\alpha})}{2(\alpha - 1)}  + \frac{1}{2} \frac{\alpha (\mu_1 - \mu_2)^2}{(\sigma^2)^*_\alpha}, \nonumber
\end{equation}
for $(\sigma^2)^*_\alpha = \alpha \sigma^2_2 + (1-\alpha)\sigma_1^2$~\cite[p. 45]{Liese_Vajda_1987}; see also~\cite{gil2013renyi}. For $\mu_1 \neq \mu_2$, $D_{1/\alpha}(P||Q) \rightarrow \infty$ as $\alpha \rightarrow 0$, while $\lim_{\alpha \rightarrow 0} D_{2-\alpha} (P||Q) = D_{2}(P||Q) < \infty$. Hence, $\text{LB}_1^\alpha$ matches the asymptotic behaviour of $L(t)$ as $\alpha \rightarrow 0$, while $\text{LB}_2^\alpha$ does not. As discussed in Section~\ref{sec:lowerbounds}, the lesser constant term in $\text{LB}_1^\alpha$ compared to $\text{LB}_2^\alpha$ means that the lower bound in $\text{LB}_2^\alpha$ is tighter for $\alpha$ close to one and small $D_\alpha(P||Q)$; this can be seen in Fig.~\ref{fig:bounds}. 

In some special cases, one can directly compute the distribution of $K$ and thus $H_\alpha(K)$, and then bound the cost $L(t)$ as $L(t) < H_\alpha(K) + 1$. Indeed,~\eqref{eq:distributionofK} gives the explicit formula $K \mid \qty{U_K = u} \sim \text{Geo}(\beta(u))$, with $\beta(u) = \mathbb{E}_{U \sim Q}\left[ \max \qty{\frac{\dd P}{\dd Q}(u), \frac{\dd P}{\dd Q}(U)} \right]^{-1}$. When $P = \mathcal{N}(\mu_1, \sigma_1)$ and $Q = \mathcal{N}(\mu_2, \sigma_2)$ are normal distributions with respective densities $p$ and $q$, we can write for $u \in \mathbb{R}$ that
\begin{align}
    \beta(u)^{-1} &= \mathbb{E}_{U \sim Q}\left[ \max \qty{\frac{\dd P}{\dd Q}(u), \frac{\dd P}{\dd Q}(U)} \right] \nonumber \\
    &= \int_{-\infty}^\infty \frac{\dd P}{\dd Q} (\min \qty{x ,u}) q(x) \dd x  \nonumber \\
    &= \int_{-\infty}^u \frac{\dd P}{\dd Q}(x) q(x) \dd x + \int_u^\infty \frac{\dd P}{\dd Q}(u) q(x) \dd x \nonumber \\
    &= \int_{-\infty}^u p(x) \dd x + \frac{\dd P}{\dd Q}(u) \int_u^\infty q(x) \dd x \nonumber \\
    &= \Phi_P(u) + \frac{\dd P}{\dd Q}(u) \left( 1 - \Phi_Q(u) \right), \nonumber
\end{align}
with $\Phi_P$ and $\Phi_Q$ the cumulative distribution functions of $P$ and $Q$, respectively. We then obtain that
\begin{equation}
    \mathbb{P}(K = k) = \mathbb{E}_{U_K \sim P} \left[ \mathbb{P}(K = k \mid U_K) \right] = \int_{-\infty}^\infty (1 - \beta(u))^{k-1} \beta(u) \; p(u) \; \dd u. \label{eq:distKnormal}
\end{equation}
Eq.~\eqref{eq:distKnormal} gives an explicit formula for the distribution of $K$, which can be computed using numerical quadrature. From this expression, we can closely approximate the R\'enyi entropy $H_\alpha(K)$ as
\begin{equation}
    H_\alpha(K) \approx \frac{1}{1-\alpha} \log( \sum_{k=1}^N \mathbb{P}(K=k)^\alpha ), \label{eq:renyiapproximation}
\end{equation} 
for some large $N$. Fig.~\ref{fig:plotwithrenyientropy} compares the bounds with $H_\alpha(K)+1$, approximated as in~\eqref{eq:renyiapproximation} with $N=1000$.\footnote{The error in this approximation of $H_\alpha(K)$, for $P = \mathcal{N}(0,1)$ and $Q = \mathcal{N}(1, 1)$, is less than $0.02$ bits for all values of $\alpha$.} 
\begin{figure*}[ht]
    \centering
    \includegraphics{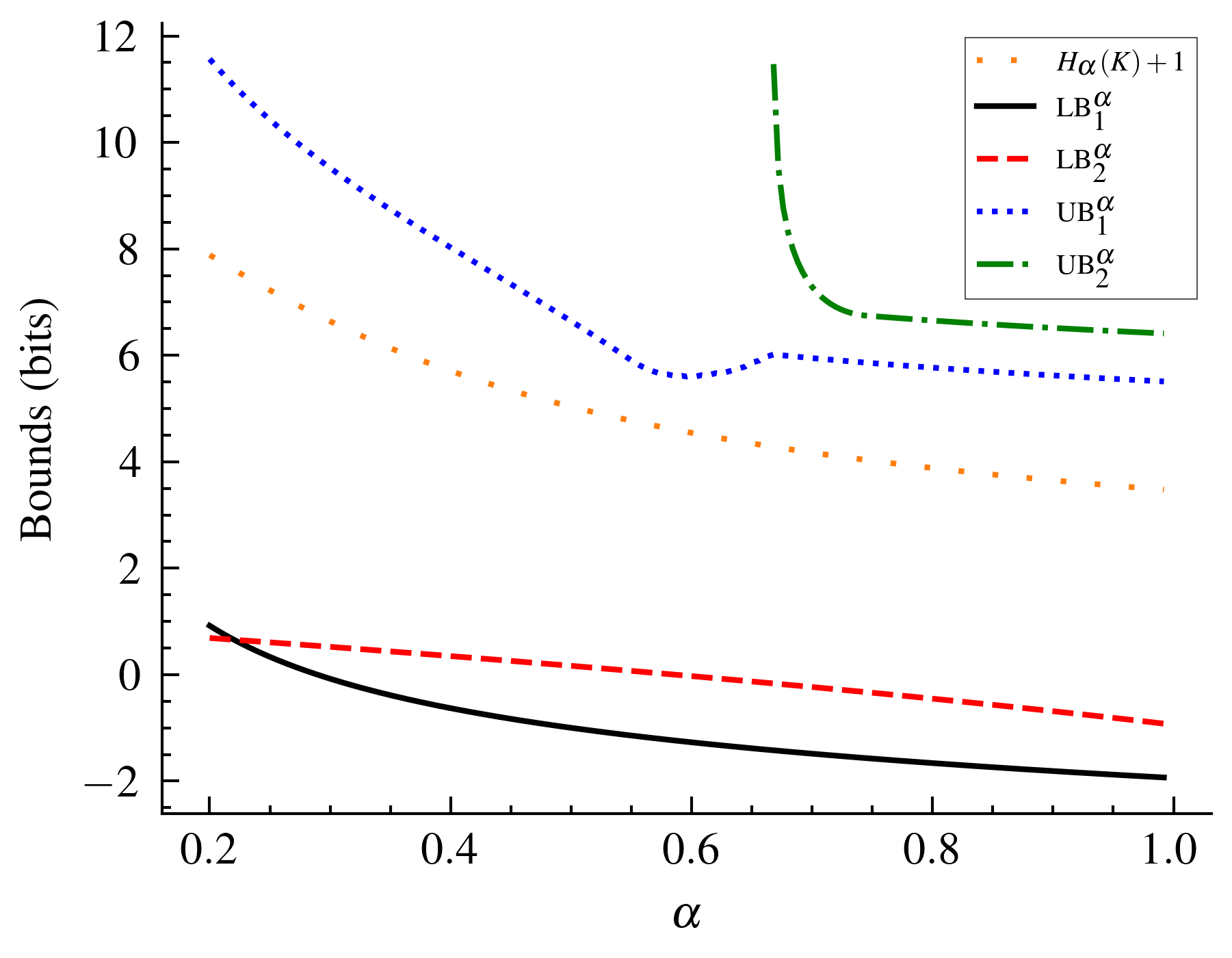}
    \caption{Comparison of the bounds on the Campbell cost and $H_\alpha(K)$ for $P =\mathcal{N}(0, 1)$ and $Q=\mathcal{N}(1, 1)$, for $0.2 < \alpha < 1$.}
    \label{fig:plotwithrenyientropy}
\end{figure*}
As seen in Fig.~\ref{fig:plotwithrenyientropy}, the obtained upper bound is tighter than $\text{UB}_1^\alpha$ or $\text{UB}_2^\alpha$. However, this is a special case where $P$ and $Q$ are close together, and therefore the probability distribution of $K$ is concentrated on low values. In this example, $1 - \sum_{k=1}^{1000} \mathbb{P}(K=k) < 2 \times 10^{-8}$, whereas for $P=\mathcal{N}(0, 1)$ and $Q=\mathcal{N}(5,1)$, $1 - \sum_{k=1}^{1000} \mathbb{P}(K=k) > 0.83$. The distribution of $K$ becomes exceedingly heavy-tailed, making accurately estimating $H_\alpha(K)$ computationally intractable. In contrast, $\text{UB}_1^\alpha$ and $\text{UB}_2^\alpha$ can be computed instantaneously, even for complex distributions that do not admit a simple form of $\mathbb{P}(K=k)$. 

As another numerical example, let $P, Q$ be Laplacian distributions, i.e., for $\theta_1, \theta_2 \in \mathbb{R}$ and $\lambda_1, \lambda_2 > 0$ we have $P = \text{Laplace}(\theta_1, \lambda_1)$ and $Q = \text{Laplace}(\theta_2, \lambda_2)$, where $\text{Laplace}(\theta, \lambda)$ denotes a Laplacian distribution with mean $\theta$ and variance $2\lambda^2$. For $\alpha > 0$ with $\alpha \neq \frac{\lambda_1}{\lambda_1 + \lambda_2}$,~\cite{gil2013renyi} gives the formula for $D_\alpha(P||Q)$ as
\begin{equation}
    D_\alpha(P||Q) = \begin{cases}
        \ln \frac{\lambda_2}{\lambda_1} + \frac{1}{\alpha - 1} \ln( \frac{\lambda_1 \lambda_2^2 g(\alpha)}{\alpha^2 \lambda_2^2 - (1-\alpha)^2 \lambda_1^2}), &\alpha \lambda_2 + (1-\alpha)\lambda_1 > 0, \\
        \infty, &\alpha \lambda_2 + (1-\alpha)\lambda_1 \leq 0,
    \end{cases} \label{eq:laplaciandivergence}
\end{equation}
where 
\begin{equation}
    g(\alpha) \coloneqq \frac{\alpha}{\lambda_1} \exp(- \frac{(1-\alpha) \lvert \theta_1 - \theta_2 \rvert}{\lambda_2}) - \frac{1-\alpha}{\lambda_2} \exp( -\frac{\alpha\lvert \theta_1 - \theta_2 \rvert}{\lambda_1} ). \nonumber
\end{equation}
Fig.~\ref{fig:laplacianbounds} compares the bounds of Theorems~\ref{thm:lowerbound} to~\ref{thm:worseupperbound} for the Laplacian distributions $P$ and $Q$. 
\begin{figure*}[ht]
    \centering
    \begin{subfigure}{0.32\textwidth}
        \includegraphics[width=\textwidth]{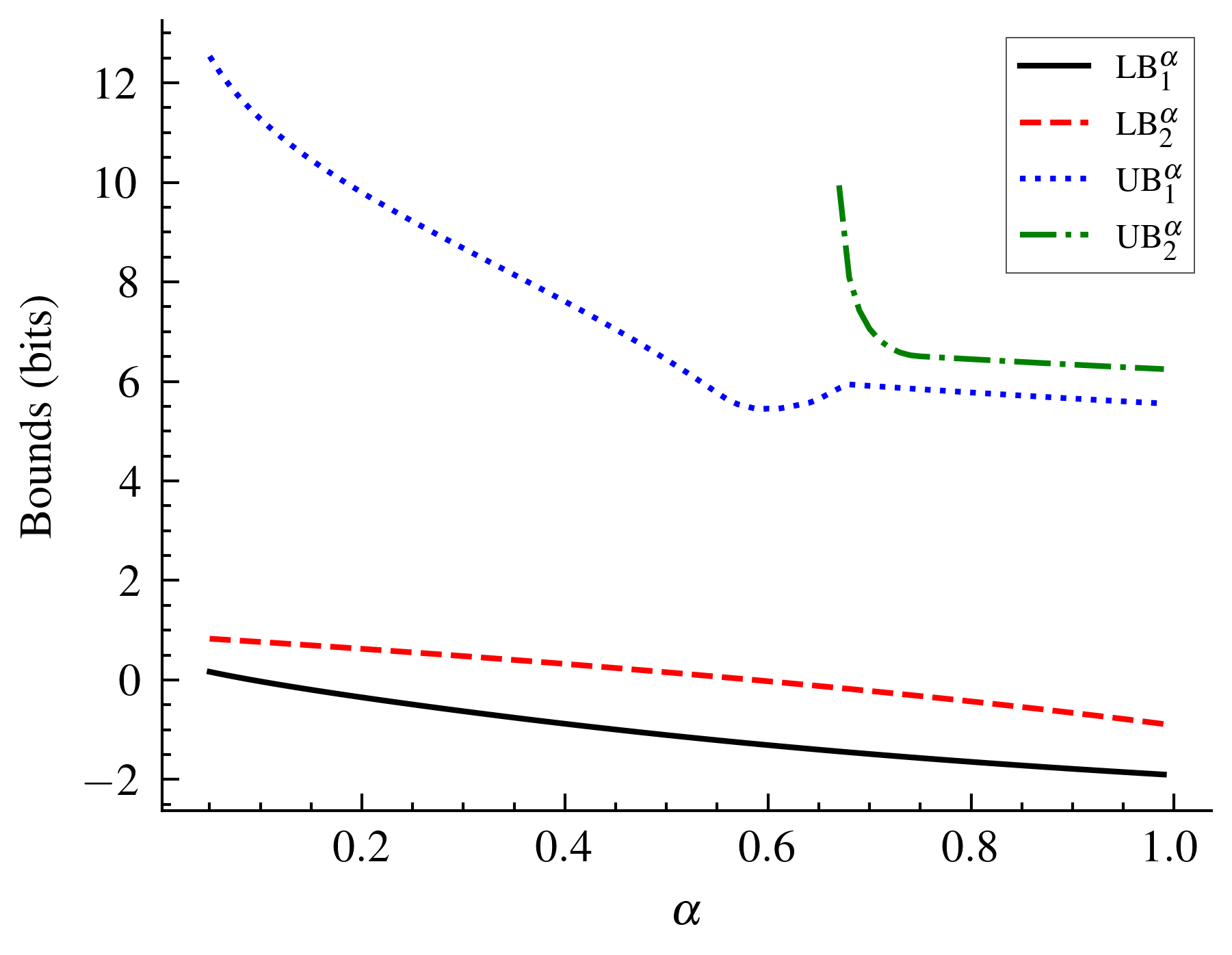}
        \caption{$\theta_1 = 0, \lambda_1 = 1$ and $\theta_2 = 1, \lambda_2 = 1$}
    \end{subfigure}
    \hfill
    \begin{subfigure}{0.32\textwidth}
        \includegraphics[width=\textwidth]{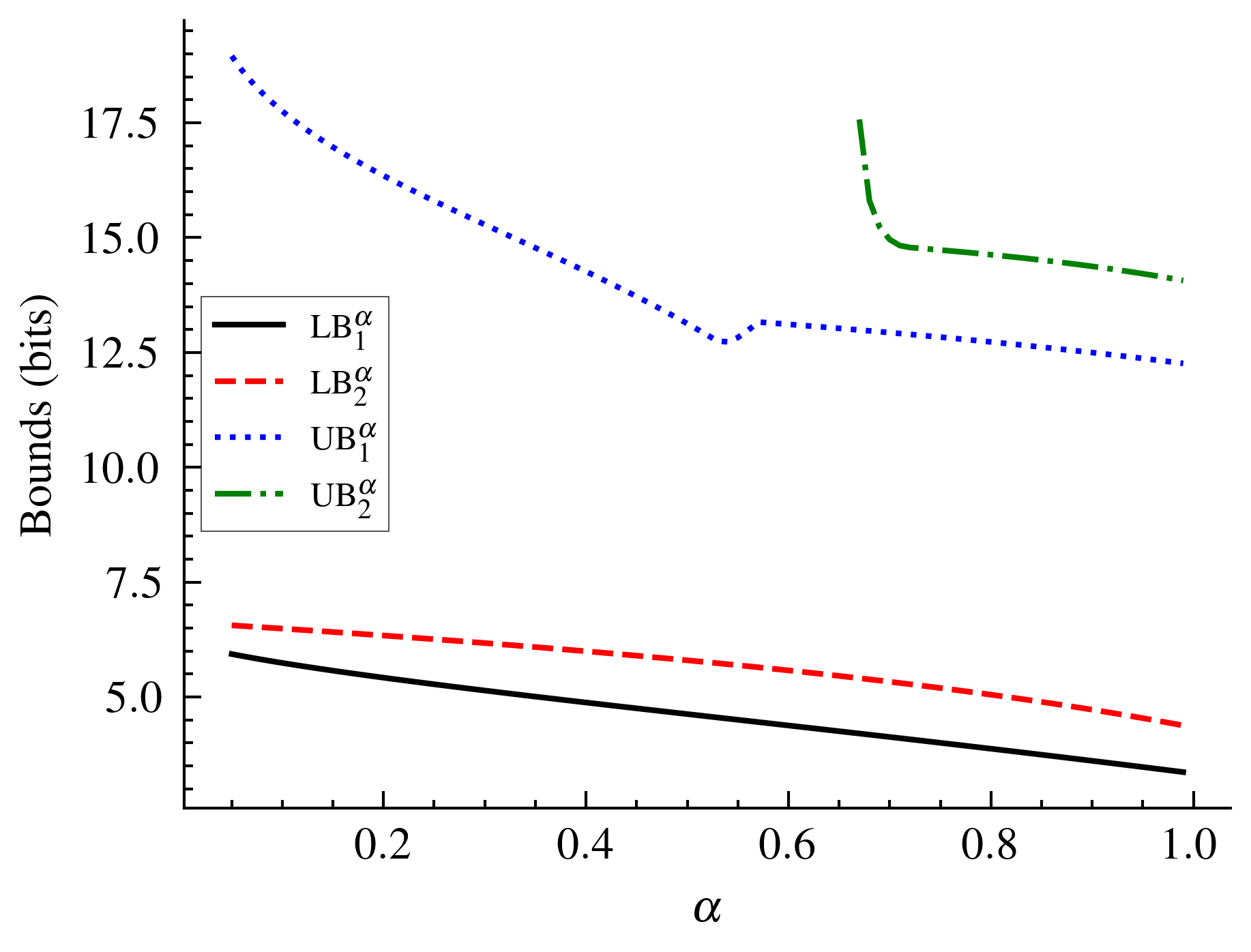}
        \caption{$\theta_1 = 0, \lambda_1 = 1$ and $\theta_2 = 5, \lambda_2 = 1$}
    \end{subfigure}
    \hfill
    \begin{subfigure}{0.32\textwidth}
        \includegraphics[width=\textwidth]{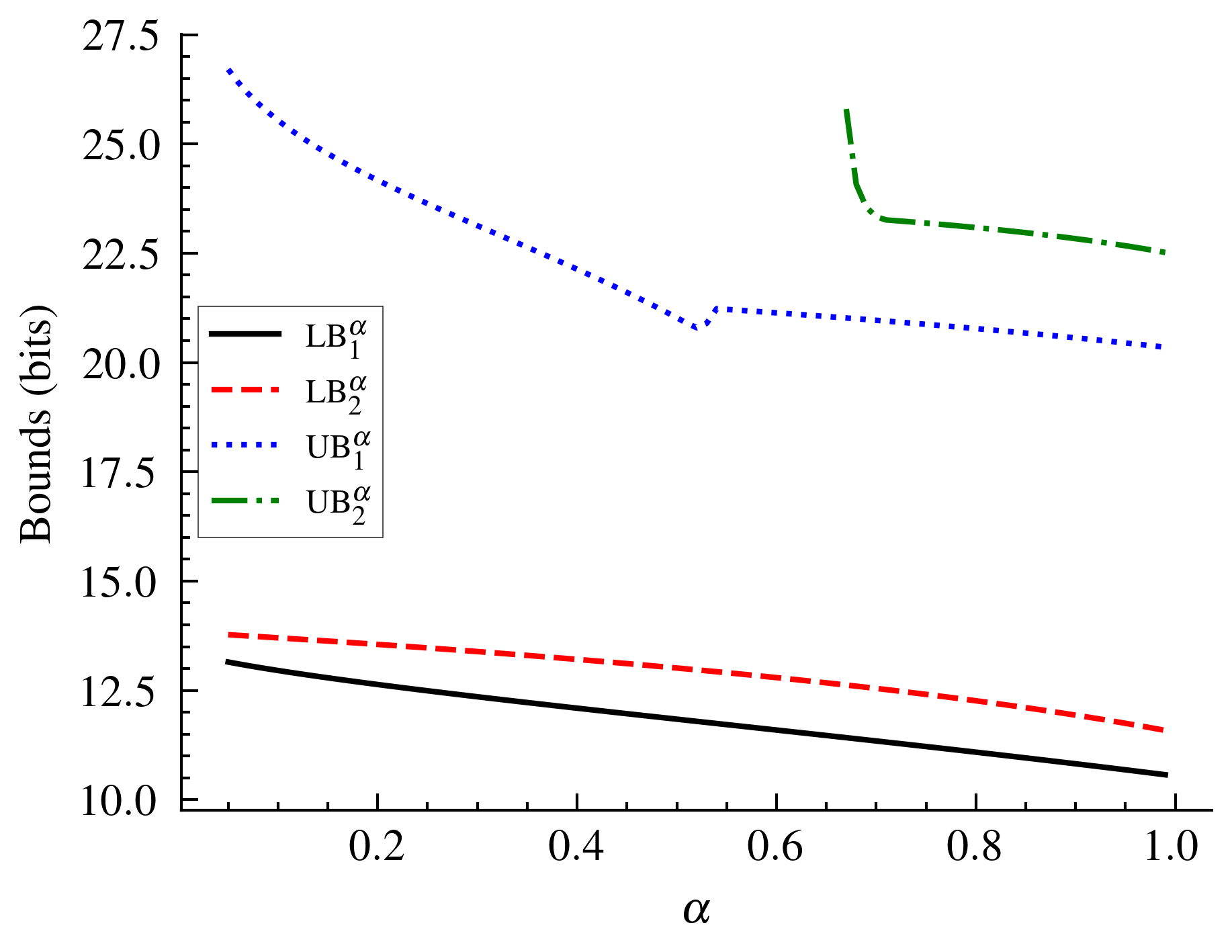}
        \caption{$\theta_1 = 0, \lambda_1 = 1$ and $\theta_2 = 10, \lambda_2 = 1$}
    \end{subfigure}
    \caption{Comparison of the bounds on the Campbell cost for $P = \text{Laplace}(\theta_1, \lambda_1)$ and $Q = \text{Laplace}(\theta_2, \lambda_2)$, for $0.05 < \alpha < 1$.}
    \label{fig:laplacianbounds}
\end{figure*}
Observe that, in each case, $\text{LB}_2^\alpha$ is tighter, despite having a leading R\'enyi divergence term of $D_{2-\alpha}(P||Q)$. 
As seen in~\eqref{eq:laplaciandivergence}, the value of $\alpha$ has less impact on its R\'enyi divergence (than in the case of normal distributions), and therefore, the larger constant term makes $\text{LB}_2^\alpha$ tighter than $\text{LB}_1^\alpha$. Moreover, in contrast to when $P$ and $Q$ are normal distributions, for $P = \text{Laplace}(\theta_1, \lambda_1)$ and $Q = \text{Laplace}(\theta_2, \lambda_2)$, if $\lambda_2 \geq \lambda_1$ then $\lim_{\alpha \to \infty} D_\alpha(P||Q) < \infty$. Hence, the asymptotics of $\text{LB}_1^\alpha$ and $\text{LB}_2^\alpha$ align more closely as $\alpha \to 0$ than in the case of Gaussians. In all three cases in Fig.~\ref{fig:laplacianbounds}, observe that $\text{LB}_2^\alpha$ and $\text{UB}_1^\alpha$ are still within 5-10 bits for most $\alpha$, diverging only as $\alpha \to 0$. 
\section{Conclusion}
In this paper, we have generalized bounds on the expected communication cost of one-shot exact sampling to the Campbell cost and R\'enyi's entropy. Such bounds on the exponential cost are useful in situations where one wishes to disproportionally penalize long codewords, such as applications with buffer overflow. Several interesting directions for future work stem from these results. Most obviously, it is an open question if $\text{UB}_1^\alpha$ can be tightened to be in the spirit of~\cite{sfrl}, specifically so that~\eqref{eq:thmupperbound} reduces to $\mathbb{E}[l(\mathcal{C}(K))] \leq D(P||Q) + \log (D(P||Q) + 1) + O(1)$ as $\alpha \to 1$. We also conjecture that the asymptotic gap between causal and noncausal samplers shown in Section~\ref{sec:asymptotics} still holds in the one-shot case and without the requirement of a nonincreasing probability distribution. In the broader field of channel simulation, the upper bound in Corollary~\ref{cor:channelsimulation} suggests the question of whether $I^{\text{\scriptsize s}}_{1/\alpha}(X; Y)$ is a lower bound on the Campbell cost of \textit{any} channel simulation algorithm. While most state-of-the-art channel simulation algorithms rely on sampling, there are other, fundamentally different, protocols~\cite{sriramu2024fast}. Proving such a lower bound would characterize the optimal asymptotic Campbell rate of channel simulation under known and worst-case input. More generally, the strong functional representation lemma has been applied to several problems outside of one-shot exact channel simulation, most notably one-shot variable-length lossy source coding and multiple description coding~\cite{sfrl}. It would be interesting to see if the strong functional representation lemma can play a role in information-theoretic coding problems using the Campbell cost or R\'enyi entropy.

\section*{Acknowledgment}
The authors would like to thank the reviewers for their helpful comments and feedback.
\appendix{}
\subsection{Proof of Lemma~\ref{lemma:alphamomentlowerbound}}
\label{app:kmomentproof}
\begin{proof}[\nopunct] We first let $g(k) = k$, i.e., we show that 
\begin{equation}
    E[K^\alpha] \geq \frac{1}{1+\alpha} 2^{\alpha D_{\alpha + 1}(P||Q)}. \label{eq:alphamomentlowerbound}
\end{equation}
To do so, we will first show that for any $k \in \mathbb{N}$, $\mathbb{P}(K = k \mid U_K = u) \leq \frac{1}{\frac{\dd P}{\dd Q}(u)}$ for $P$-almost every $u \in \mathcal{U}$. Define $\nu$ to be the probability measure generated by the joint distribution of $K$ and $U_K$, i.e., $\nu(A) \coloneqq \mathbb{P}((K, U_K) \in A)$ for each Borel set $A \subseteq \mathbb{N} \times \mathcal{U}$. Then, defining $\nu_k(B) \coloneqq \nu(k \times B)$, for any Borel set $B \subseteq \mathcal{U}$ we get that 
\begin{align*}
    \int_B 1 \dd Q(u) = Q(B) &= \sum_{k=1}^\infty \nu_k(B) \\
    &\geq \nu_k(B) \\
    &= \int_B \mathbb{P}(K = k \mid U_K = u) \dd P(u) \\
    &= \int_B \mathbb{P}(K = k \mid U_K = u) \frac{\dd P}{\dd Q}(u) \dd Q(u).
\end{align*}
    Thus,
    \begin{equation}
        \mathbb{P}(K = k \mid U_K = u) \leq \frac{1}{\frac{\dd P}{\dd Q}(u)} \label{eq:conditionalboundpk}
    \end{equation}
    for $P$-almost every $u \in \mathcal{U}$. Let $E \subseteq \mathcal{U}$ be the set of all $u$ where \eqref{eq:conditionalboundpk} holds. For $u \in E$, writing $p_u(k) = \mathbb{P}(K = k | U_K = u)$ and $c_u = (\dd P / \dd Q)(u)$, we have that $p_u(k) \leq \frac{1}{c_u}$ for all $k$. Then, for $\qty{c_u} \coloneqq c_u - \lfloor c_u \rfloor$ the fractional part of $c_u$, we claim
    \begin{equation}
        \sum_{k=1}^\infty k^\alpha p_u(k) \geq \frac{1}{c_u} \left( \sum_{k=1}^{\lfloor c_u \rfloor } k^\alpha + \{c_u\}(\lceil c_u \rceil)^\alpha \right) \label{eq:redistribution}.
    \end{equation}
    If $p_u(k)$ does not have distribution $p_u(k) = \frac{1}{c_u}$ for $k \in \qty{1, \ldots, \lfloor c_u \rfloor}$ and $p_u(\lceil c_u \rceil) = \frac{\qty{c_u}}{c_u}$, we can move any excess probability from indices $k \geq \lceil c_u \rceil$ until it has such a distribution. Since we shift probabilities from indices $k \geq \lceil c_u \rceil$ to $k \leq \lceil c_u \rceil$ and $k^\alpha$ is increasing in $k$,~\eqref{eq:redistribution} must hold. Hence, for all $u \in E$,
    \begin{align}
        \mathbb{E}\left[K^\alpha \mid U_K = u \right] = \sum_{k=1}^\infty k^\alpha p_u(k) &\geq \frac{1}{c_u} \left( \sum_{k=1}^{\lfloor c_u \rfloor } k^\alpha + \{c_u\}(\lceil c_u \rceil)^\alpha \right) \nonumber \\
        &\geq \frac{1}{c_u}\left( \int_0^{c_u} t^\alpha \dd t \right) = \frac{1}{1 + \alpha} c_u^{\alpha} = \frac{1}{1 + \alpha} \left[ \frac{\dd P}{\dd Q}(u) \right]^\alpha, \label{eq:integrallowerbound}
    \end{align}
    where the inequality in~\eqref{eq:integrallowerbound} follows because $f(t) = t^\alpha$ is an increasing function of $t$. Taking expectation over $U \sim P$ gives, as $P(E) = 1$, the lower bound in~\eqref{eq:alphamomentlowerbound}. 
    
    For an arbitrary bijection $g : \mathbb{N} \rightarrow \mathbb{N}$, note that for every $k \in \mathbb{N}$,~\eqref{eq:conditionalboundpk} implies
    \begin{equation}
        \mathbb{P}(g(K) = k \mid U_K = u) = \mathbb{P}(K = g^{-1}(k) \mid U_K = u) \leq \frac{1}{\frac{\dd P}{\dd Q}(u)}, \nonumber
    \end{equation}
    for $P$-almost every $u \in \mathcal{U}$. Following the same steps as above, we can conclude that $\mathbb{E}[(g(K))^\alpha] \geq \frac{1}{1+\alpha} 2^{\alpha D_{\alpha+1}(P||Q)}$.
\end{proof} 

\subsection{Proof of Proposition~\ref{prop:simpleconstant}} \label{app:simpleconstant}
\begin{proof}[\nopunct]
We aim to show that, for any $\epsilon > 0$ and $0 < \alpha < 1$, 
\begin{equation}
    c_1(\alpha, \epsilon) < \left( \frac{1}{2(1-\alpha)} + \frac{1}{2} + \epsilon \right) \log(\frac{1}{\alpha}) + \log(\frac{1}{\epsilon}) + 1.5\epsilon^2 + 4.5 \epsilon + 2.6.
\end{equation}
First let $1/2 < \alpha < 1$ and $\epsilon < \frac{2\alpha - 1}{1-\alpha}$ (we are in the first branch of $c_1(\alpha, \epsilon)$) and note that
\begin{align}
    (1 + \epsilon) \log e + 1 + \log (1 + \frac{1}{\epsilon}) &= \log(\frac{1}{\epsilon}) + \log(\epsilon + 1) + (1 + \epsilon) \log e+ 1 \nonumber \\
    &\leq \log(\frac{1}{\epsilon}) + 2\epsilon \log e + \log e+ 1 \label{eq:inequalityhere} \\
    &< \log (\frac{1}{\epsilon}) + 3\epsilon + 2.6,
\end{align}
where~\eqref{eq:inequalityhere} follows from $\log (x + 1) \leq x \log e$ for any $x > 0$. As $\epsilon^2 > 0$ and $\log(\frac{1}{\alpha}) > 0$ for $\alpha > 1/2$, the proposition holds for this case. \\
Let now $0 < \alpha < 1/2$ or $\epsilon \geq \frac{2\alpha - 1}{1-\alpha}$. Let $x = \frac{1 + \epsilon(1-\alpha)}{\alpha}$. Note that the domain assumption on $\alpha$ and $\epsilon$ guarantees $x \geq 2$. Using the Taylor series bound $\log \Gamma(x) < (x - 1/2)\log x + (1-x)\log e$, valid for $x \geq 1$, we obtain that
\begin{align}
     \frac{\alpha}{1-\alpha}\log \Gamma(x) &< \frac{\alpha}{1-\alpha} \left[ \left( \frac{1}{\alpha} + \epsilon \frac{1-\alpha}{\alpha} - \frac{1}{2}\right)  \log(\frac{1 + \epsilon(1-\alpha)}{\alpha}) + \left(1 - \frac{1}{\alpha} - \epsilon \frac{1-\alpha}{\alpha} \right) \log e\right] \nonumber \\
     &= \left( \frac{1 + 1 - \alpha}{2(1-\alpha)} + \epsilon \right)  \left( \log(\frac{1}{\alpha}) + \log(1 + \epsilon(1-\alpha)) \right) + \left(\frac{\alpha}{1-\alpha} - \frac{1}{1-\alpha} - \epsilon \right) \log e \nonumber \\
    &\leq \left( \frac{1}{2(1-\alpha)} + \frac{1}{2} + \epsilon \right) \log(\frac{1}{\alpha}) + \left( \frac{1}{2(1-\alpha)} + \frac{1}{2} + \epsilon \right) \epsilon(1-\alpha)\log e -(1+ \epsilon) \log e 
    \label{eq:loginequalitystep} \\
    &= \left( \frac{1}{2(1-\alpha)} + \frac{1}{2} + \epsilon \right) \log(\frac{1}{\alpha}) - \log e - \frac{\alpha \epsilon}{2} \log e + (1-\alpha)\epsilon^2 \log e. \nonumber
\end{align}
where~\eqref{eq:loginequalitystep} again follows from $\log(x + 1) \leq x \log e$. Combining this bound with the remaining terms of $c_1(\alpha, \epsilon)$ gives, after collecting like terms,
\begin{align}
    c_1(\alpha, \epsilon) &< \left( \frac{1}{2(1-\alpha)} + \frac{1}{2} + \epsilon \right) \log(\frac{1}{\alpha}) + \log (\frac{1}{\epsilon}) + \left( 4 - \log e - \frac{2\alpha}{1-\alpha} \right) + \epsilon \left( 3 + \left(1 - \frac{\alpha}{2} \right) \log e \right) + \epsilon^2 (1-\alpha) \log e \nonumber  \\
    &< \left( \frac{1}{2(1-\alpha)} + \frac{1}{2} + \epsilon \right) \log(\frac{1}{\alpha}) + \log (\frac{1}{\epsilon}) + 1.5 \epsilon^2  + 4.5 \epsilon + 2.6. \nonumber
\end{align}
The result follows.
\end{proof}
\subsection{Proof of Lemma~\ref{lemma:geometricmoment}}
\label{app:lemmageometricproof}
\begin{proof}[\nopunct]
The case $r = 1$ holds trivially, as $\mathbb{E}[X] = \frac{1}{p}$. Let $r > 1$. We know that $X = \lceil Y \rceil$, where $Y$ is an exponential random variable with parameter $\lambda = -\ln(1-p)$. Thus, $X \leq Y + 1$ and we have
\begin{equation*}
    \mathbb{E}[Y^r] = \int_0^\infty y^r \lambda e^{-\lambda y} \dd y = \int_0^\infty \left( \frac{x}{\lambda} \right)^r e^{-x} \dd x = \frac{\Gamma(r + 1)}{\lambda^r}. 
\end{equation*}
Further, note that for any $r > 1$ and $x > 0$, $(x + 1)^r \leq 2^{r-1} (x^r + 1)$. We can therefore upper bound $\mathbb{E}[X^r]$ as 
\begin{align}
    \mathbb{E}[X^r] \leq \mathbb{E}[(Y + 1)^r] &\leq 2^{r-1} \left( \mathbb{E}[Y^r] + 1 \right) \nonumber \\
    &= 2^{r-1} \left( \frac{\Gamma(r + 1)}{\lambda^r} + 1\right) \nonumber \\
    &= 2^{r-1} \left( \frac{\Gamma(r + 1)}{\left[ -\ln(1-p) \right]^r} + 1\right) \nonumber \\
    &\leq 2^{r-1} \left( \frac{\Gamma(r+1)}{p^r} + 1 \right), \label{eq:lnbound}
\end{align}
where~\eqref{eq:lnbound} follows because $p \leq -\ln(1-p)$ for any $0 < p < 1$ and $x \mapsto x^r$ is an increasing function. 
\end{proof}
\subsection{Constant Term Analysis}
\label{app:kraftanalysis}
We herein prove the existence of the uniquely decodable encodings of $K$ used in the proofs of Theorems~\ref{thm:upperbound} and~\ref{thm:worseupperbound}. First, in Theorem~\ref{thm:upperbound}, we claim that there exists a prefix code with codeword lengths $n_k = (1 + \epsilon) \log k + c_\epsilon$, for any $\epsilon > 0$ and $c_\epsilon \leq 1 + \log (1 + \frac{1}{\epsilon})$ a constant such that $n_k$ is an integer and the code is a valid prefix code. The existence of such a code is guaranteed by the Kraft inequality, since for all $\epsilon > 0$ there exists $c_\epsilon$ such that
\begin{equation}
    \sum_{k=1}^\infty \frac{1}{k^{(1+\epsilon)}} \leq 2^{c_\epsilon}. \label{eq:kraftinequality}
\end{equation}
We can upper bound the sum in~\eqref{eq:kraftinequality} by
\begin{equation}
    \sum_{k=1}^\infty \frac{1}{k^{(1+\epsilon)}} \leq 1 + \int_{1}^\infty x^{-(1+\epsilon)} \dd x = 1 + \frac{1}{\epsilon}. \nonumber
\end{equation}
Returning to the requirement that $n_k$ be an integer, we can take the ceiling of our codeword length, imposing at most a $+1$ penalty, to get that
\begin{equation}
    c_\epsilon \leq 1 + \log(1 + \frac{1}{\epsilon}), \nonumber
\end{equation}
as claimed. Similarly, in Theorem~\ref{thm:worseupperbound} we consider a universal code with codeword lengths $n_k = \log k  + (1 + \epsilon)\log \log (k+1) + c_\epsilon$, for any $\epsilon > 0$ and $c_\epsilon \leq 1 + \log (\frac{\ln(2)}{\epsilon} + \frac{3}{2})$. By the Kraft inequality, we require that
\begin{equation}
    1 \geq \sum_{k=1}^\infty 2^{-n_k} = \sum_{k=1}^\infty \frac{1}{k\log^{(1+\epsilon)}(k+1)2^{c_\epsilon}}. \nonumber
\end{equation}
Hence, setting 
\begin{equation}
    c_\epsilon = \log (\sum_{k=1}^\infty \frac{1}{k\log^{(1+\epsilon)}(k+1)}) < \infty, \label{eq:constanttermsum}
\end{equation}
the code satisfies the Kraft inequality with equality. We can upper bound the sum in~\eqref{eq:constanttermsum} by
\begin{align}
    &\sum_{k=1}^\infty \frac{1}{k\log^{(1+\epsilon)}(k+1)}  \nonumber \\
    &= \sum_{k=1}^\infty \left[ \frac{1}{(k+1) \log^{\epsilon + 1}(k + 1)} + \frac{1}{k(k+1)\log^{1+\epsilon}(k+1)} \right]. \nonumber \\
    &\leq \frac{1}{2} + \int_{1}^\infty \frac{1}{(x+1)\log^{(1+\epsilon)}(k+1)} \dd x + \sum_{k=1}^\infty \frac{1}{k(k+1)} \nonumber \\
    &\leq \frac{1}{2} + \frac{\ln(2)}{\epsilon} + 1 = \frac{3}{2} + \frac{\ln(2)}{\epsilon}. \nonumber
\end{align}
Again imposing at most a $+1$ penalty to ensure $n_k$ is an integer, we obtain that
\begin{equation}
    c_\epsilon \leq 1 + \log(\frac{3}{2} + \frac{\ln(2)}{\epsilon}). \nonumber
\end{equation}
\subsection{Proof of Theorem~\ref{thm:causallowerbound}} \label{app:causalasymptoticproof}
\begin{proof}[\nopunct]
    First assume that $t < 1$ (consequently $\alpha \in (1/2, 1)$). Let $K_n$ be the communicated index when the proposal and target distributions are $P^{\otimes n}$ and $Q^{\otimes n}$, respectively. We first wish to prove that for any $\epsilon > 0$,
    \begin{equation}  
        \liminf_{n \to \infty}\frac{1}{n t} \log( \mathbb{E}[K_n^t]) \geq D_{\frac{1}{1-t}}(P||Q) - \frac{2\epsilon}{t}. \label{eq:epsilontozero}
    \end{equation}
    
    Our proof follows closely~\cite{liu2018rejection}, but contains details omitted there. Define the tilted distribution $R_n$ by
    \begin{align*}
        \frac{\dd R_n}{\dd Q^{\otimes n}} &\coloneqq \frac{\exp(\frac{1}{1-t} \log \frac{\dd P^{\otimes n}}{\dd Q^{\otimes n}})}{\int \exp(\frac{1}{1-t} \log \frac{\dd P^{\otimes n}}{\dd Q^{\otimes n}}) \dd Q^{\otimes n}}.
    \end{align*}
    Define the random variable $X_n = \log \frac{\dd P^{\otimes n}}{\dd Q^{\otimes n}}(Z_n) = \sum_{i = 1}^n \log \frac{\dd P}{\dd Q}(Z_i)$, for $Z_n \sim R_n$ and $Z \sim R$ (the one-dimensional tilted distribution). Fix $\epsilon > 0$ and define
    \begin{align*}
        \tau_n &\coloneqq \mathbb{E}[X_n] - \epsilon n; \\
        q_n &\coloneqq \int \exp(\frac{1}{1-t} \log \frac{\dd P^{\otimes n}}{\dd Q^{\otimes n}}) \dd Q^{\otimes n} = \exp( \frac{t}{1-t} D_{\frac{1}{1-t}}(P^{\otimes n} || Q^{\otimes n})).
    \end{align*}
    Then,
    \begin{align*}
        Q^{\otimes n} \left( \qty{\tau_n < \log \frac{\dd P^{\otimes n}}{\dd Q^{\otimes n}}} \right) &\geq Q^{\otimes n} \left( \qty{\tau_n <  \log \frac{\dd P^{\otimes n}}{\dd Q^{\otimes n}} < \tau_n + 2 \epsilon n }\right) \\
        &\geq q_n \exp(- \frac{\tau_n + 2\epsilon n}{1-t}) R_n \left( \qty{ \tau_n < \log \frac{\dd P^{\otimes n}}{\dd Q^{\otimes n}} < \tau_n + 2\epsilon n }\right) \\
        &= q_n \exp(- \frac{\tau_n + 2\epsilon n}{1-t}) \mathbb{P}\left( \lvert X_n - \mathbb{E}[X_n] \rvert < \epsilon n \right). \label{eq:probability}
    \end{align*}
    Theorem 3 in~\cite{liu2018rejection} proves that 
    \begin{equation*}
        \mathbb{E}[K_n^t] \geq  (1-t) t^t \sup_{\tau \in \mathbb{R}} \exp(\tau) \left[ Q^{\otimes n} \left(\log \frac{\dd P^{\otimes n}}{\dd Q^{\otimes n}} > \tau \right) \right]^{1-t}.
    \end{equation*}
    Choosing $\tau = \tau_n$, noting that we require $ \mathbb{E}[ \lvert X_n \rvert] < \infty$ (which will be established subsequently), we obtain that
    \begin{align*}
        \mathbb{E}[K_n^t] &\geq (1-t)t^t \exp(\tau_n) q_n^{1-t} \exp(- \frac{\tau_n + 2\epsilon n}{1-t} \cdot 1-t) \left[ \mathbb{P}(\lvert X_n - \mathbb{E}[X_n] \rvert < \epsilon n) \right]^{1-t} \\
        &= (1-t)t^t \exp(t D_{\frac{1}{1-t}}(P^{\otimes n} || Q^{\otimes n})) \exp(-2 \epsilon n) \left[ \mathbb{P}(\lvert X_n - \mathbb{E}[X_n] \rvert < \epsilon n) \right]^{1-t}. 
    \end{align*}
        Therefore, taking logarithms and dividing by $n t$ we get that (using the additivity of the R\'enyi divergence for product distributions)
        \begin{align}
            \frac{1}{nt}\log(\mathbb{E}[K^t_n]) \geq D_{\frac{1}{1-t}}(P||Q) + \frac{1}{nt} \log ((1-t)t^t) - \frac{2\epsilon}{t} + \frac{1-t}{t} \log (\mathbb{P}( \lvert X_n - \mathbb{E}[X_n] \rvert < \epsilon n)). 
            \label{eq:firstprobability}
        \end{align}
    We can rewrite the probability in~\eqref{eq:firstprobability} as 
    \begin{align*}
        \mathbb{P} \left( \left\lvert \sum_{i = 1}^n \log \frac{\dd P}{\dd Q}(Z_i) - n\mathbb{E}_{Z \sim R}\left[\log \frac{\dd P}{\dd Q}(Z)\right] \right\rvert < \epsilon n \right)
        &= \mathbb{P} \left( \left\lvert \frac{1}{n} \sum_{i = 1}^n \log \frac{\dd P}{\dd Q}(Z_i) - \mathbb{E}_{Z \sim R}\left[\log \frac{\dd P}{\dd Q}(Z)\right] \right\rvert < \epsilon \right).
    \end{align*}
    Again with the assumption that $\mathbb{E}[ \lvert X_n \rvert] < \infty$, the weak law of large numbers (e.g.~\cite[Cor. 5.4.11]{firstlookrosenthal}) gives that the probability in~\eqref{eq:firstprobability} tends to 1;~\eqref{eq:epsilontozero} follows.
    \par
    Now, as $K_n$ has a nonincreasing probability distribution by assumption, a lower bound on its encoding cost is the optimal one-to-one code with codeword lengths $\lfloor \log(k+1) \rfloor$ for every $k \in \mathbb{N}$. Therefore,
    \begin{align}
        L^*_n(t) &\geq \frac{1}{t} \log( \sum_{k = 1}^\infty \mathbb{P}(K_n = k) 2^{t \lfloor \log(k+1) \rfloor} ) \nonumber \\
        &\geq \frac{1}{t} \log( \sum_{k = 1}^\infty \mathbb{P}(K_n = k) 2^{t (\log(k+1) - 1) }) \nonumber \\
        &\geq \frac{1}{t} \log( \mathbb{E}[K_n^t]) - 1 \label{eq:intermediatelowerbound}.
    \end{align}
    Dividing both sides by $n$ and taking $n \to \infty$ gives that, for any $\epsilon > 0$,
    \begin{equation}
        \liminf_{n \to \infty} \frac{L^*_n(t)}{n} \geq \liminf_{n \to \infty} \frac{1}{nt} \log( \mathbb{E}[K^t_n]) \geq D_{\frac{1}{1-t}}(P||Q) - \frac{2\epsilon}{t}. \nonumber
    \end{equation}
    Noting that $\alpha = \frac{1}{1+t}$ and taking $\epsilon \to 0$ yields the asymptotic inequality.
    \par 
    We now turn to the assumption that $\mathbb{E}[ \lvert X_n \rvert] < \infty$, which is equivalent to $ \mathbb{E}_{Z \sim R}[\lvert \log \frac{\dd P}{\dd Q}(Z) \rvert] < \infty$. We will show this follows from our assumption that $D_{\frac{1}{1-t} + \delta}(P||Q) < \infty$ for some $\delta > 0$. Using a change of measure, we can write
    \begin{align}
        \int \left\lvert \log \frac{\dd P}{\dd Q} \right\rvert \dd R &= \int \left\lvert \log \frac{\dd P}{\dd Q} \right\rvert \frac{\left( \frac{\dd P}{\dd Q} \right)^{\frac{1}{1-t}}}{\mathbb{E}_Q \left[\left( \frac{\dd P}{\dd Q} \right)^{\frac{1}{1-t}}\right]} \dd Q \nonumber \\
        &= \frac{\mathbb{E}_{Q} \left[ \left\lvert \log \frac{\dd P}{\dd Q} \right\rvert \left( \frac{\dd P}{\dd Q} \right)^{\frac{1}{1-t}} \right]}{\mathbb{E}_Q \left[\left( \frac{\dd P}{\dd Q} \right)^{\frac{1}{1-t}}\right]}. \label{eq:problematicdenominator}
    \end{align} 
    Note that the denominator in~\eqref{eq:problematicdenominator} causes no problems because $\mathbb{E}_Q\left[\left( \frac{\dd P}{\dd Q} \right)^{\frac{1}{1-t}}\right] = \exp(\frac{t}{1-t} D_{\frac{1}{1-t}}(P||Q)) < \infty$. We see that, for any $\delta \in (0, 1]$,
    \begin{align}
        \mathbb{E}_Q \left[\left\lvert \log \frac{\dd P}{\dd Q} \right\rvert \left( \frac{\dd P}{\dd Q} \right)^{\frac{1}{1-t}}\right] 
        &\leq \frac{1}{\delta \ln 2} \mathbb{E}_Q \left[ \left( \frac{\dd P}{\dd Q} \right)^{\frac{1}{1-t} + \delta} \mathbbm{1}_\qty{\frac{\dd P}{\dd Q} \geq 1} \right] + \frac{1}{\delta \ln 2} \mathbb{E}_Q \left[ \left( \frac{\dd P}{\dd Q} \right)^{\frac{1}{1-t}-\delta}\mathbbm{1}_\qty{0 < \frac{\dd P}{\dd Q} < 1} \right] \label{eq:sharploginequality} \\
        &\leq \frac{1}{\delta \ln 2} \left( \mathbb{E}_Q \left[\left( \frac{\dd P}{\dd Q} \right)^{\frac{1}{1-t}+\delta }\right] + 1  \right) \nonumber \\
         &= \frac{1}{\delta \ln 2} \left( \exp( \left(\delta + \frac{t}{1-t} \right) D_{\frac{1}{1-t} + \delta}(P||Q) ) \right) < \infty. \nonumber
    \end{align}
    Here,~\eqref{eq:sharploginequality} follows from the inequalities, for any $x \geq 0$ and $\delta \in (0, 1]$,
    \begin{equation*}
       ( \log x )^+ \leq \frac{x^\delta}{\delta \ln 2}\mathbbm{1}_{\qty{x \geq 1}}, \quad \quad ( \log x )^- \leq \frac{x^{-\delta}}{\delta \ln 2} \mathbbm{1}_{\qty{0 < x < 1}}.
    \end{equation*}
    \par
    Finally, for $t \geq 1$, equivalently $\alpha \in (0, 1/2]$, we use the lower bound on $\mathbb{E}[K]$ of~\cite[Thm. III.1]{goc2024channel} to obtain that
    \begin{equation}
        \mathbb{E}[K_n^t] \geq \mathbb{E}[K_n]^t \geq 2^{t D_\infty(P^{\otimes n} ||Q^{\otimes n})} = 2^{t n D_\infty(P||Q)}, \nonumber
    \end{equation}
    where the first inequality follows from Jensen's inequality. Using the same one-to-one optimal encoding of $K$ and steps as in the case $t < 1$, the result follows. 
\end{proof}

\bibliography{references.bib}
\bibliographystyle{IEEEtran}

\end{document}